\documentclass{article}
\usepackage{graphicx} 

\usepackage{geometry}
\usepackage{amsmath,amssymb,amsthm}
\usepackage[linesnumbered,noend]{algorithm2e}
\usepackage{tabularx}
\usepackage{booktabs}
\usepackage{makecell}
\usepackage[table]{xcolor}
\usepackage{caption}
\usepackage{csquotes}
\usepackage[final]{microtype}

\definecolor{paperaccent}{RGB}{178,125,20}
\definecolor{summaryrowgray}{RGB}{245,245,245}
\definecolor{algorithmbackground}{rgb}{0.95,1,1}
\usepackage{tikz}
\usepackage{titlesec}
\usepackage{fancyhdr}
\usepackage[hidelinks]{hyperref}
\usepackage[colorinlistoftodos,prependcaption,textsize=tiny]{todonotes}

\usepackage{listings}
\lstdefinestyle{algorithmstyle}{
    backgroundcolor=\color{algorithmbackground},
    columns=fixed,
    basicstyle=\ttfamily\footnotesize,
    basewidth=0.5em,
    breakatwhitespace=false,
    breaklines=true,
    captionpos=t,
    frame=single,
    keepspaces=true,
    mathescape=true,
    numbers=left,
    numberstyle=\small,
    numbersep=5pt,
    showspaces=false,
    showstringspaces=false,
    showtabs=false,
    tabsize=2
}
\DontPrintSemicolon
\SetKwInput{KwInput}{Input}
\SetKwInput{KwOutput}{Output}
\SetKwProg{Fn}{Function}{:}{}

\definecolor{resultImpossible}{RGB}{255,238,238}
\definecolor{resultLower}{RGB}{255,246,224}
\definecolor{resultKnown}{RGB}{235,244,255}
\definecolor{resultUnknown}{RGB}{238,248,240}
\definecolor{resultTerminating}{RGB}{241,238,255}

\RestyleAlgo{ruled}
\SetAlCapFnt{\small\bfseries}
\SetAlCapNameFnt{\small\bfseries}
\SetAlgoNlRelativeSize{-1}
\SetAlgoInsideSkip{medskip}
\SetKwSty{procname}
\SetFuncSty{procname}

\newtheoremstyle{lipics}
  {6pt}
  {6pt}
  {\itshape}
  {}
  {\bfseries}
  {.}
  {0.5em}
  {\ensuremath{\blacktriangleright}\ \thmname{#1}\thmnumber{ #2}\thmnote{ (#3)}}
\newtheoremstyle{lipicsdefinition}
  {6pt}
  {6pt}
  {}
  {}
  {\bfseries}
  {.}
  {0.5em}
  {\ensuremath{\blacktriangleright}\ \thmname{#1}\thmnumber{ #2}\thmnote{ (#3)}}
\theoremstyle{lipics}
\newtheorem{theorem}{Theorem}[section]

\newtheorem{lemma}[theorem]{Lemma}
\newtheorem{corollary}[theorem]{Corollary}
\theoremstyle{lipicsdefinition}

\newtheorem{remark}[theorem]{Remark}

\newcommand{\procname}[1]{\textnormal{\textsc{#1}}}
\newcommand{\ProblemName}[1]{\texorpdfstring{\textnormal{\textsc{#1}}}{#1}}

\newcommand{\algproc}[2]{\texorpdfstring{\hyperref[#1]{\procname{#2}}\xspace}{#2}}
\newcommand{\Flood}{\algproc{alg:flood}{Flood}}
\newcommand{\DistinctValues}{\algproc{alg:distinct-values}{DistinctValues}}
\newcommand{\UpperBound}{\algproc{alg:upper-bound}{UpperBound}}

\newcommand{\InputFrequency}{\algproc{alg:bc-input-frequency}{InputFrequency}}
\newcommand{\InputMultiset}{\algproc{alg:bc-input-multiset}{InputMultisetWithMultiLeaders}}
\newcommand{\ConstraintRound}{\algproc{alg:bc-constraint-round}{ConstraintRound}}

\newcommand{\AdaptiveFlooding}{\algproc{alg:bc-adaptive-flooding}{AdaptiveFlooding}}
\newcommand{\Restart}{\procname{Restart}\xspace}
\newcommand{\DetectError}{\procname{DetectError}\xspace}

\newcommand{\StabilizingInputMultiset}{\algproc{alg:bc-stabilizing-input-multiset}{StabilizingInputMultiset}}

\SetKwFunction{AlgFlood}{Flood}
\SetKwFunction{AlgDistinctValues}{DistinctValues}
\SetKwFunction{AlgUpperBound}{UpperBound}

\SetKwFunction{AlgConstraintRound}{ConstraintRound}
\SetKwFunction{AlgInputFrequency}{InputFrequency}
\SetKwFunction{AlgInputMultiset}{InputMultisetWithMultiLeaders}

\SetKwFunction{AlgAdaptiveSimulation}{SimulateWithAdaptiveFlooding}
\SetKwFunction{AlgStartTrial}{StartTrial}

\SetKwFunction{AlgAdaptiveFlooding}{AdaptiveFlooding}
\SetKwFunction{AlgInitializeCommunication}{InitializeCommunication}
\SetKwFunction{AlgRunAdaptiveLayer}{RunAdaptiveLayer}
\SetKwFunction{AlgErrorRound}{ErrorRound}
\SetKwFunction{AlgResetRound}{ResetRound}
\SetKwFunction{AlgControlRound}{ControlRound}
\SetKwFunction{AlgResetCalendar}{ResetCalendar}
\SetKwFunction{AlgDetectError}{DetectError}

\SetKwFunction{AlgStabilizingInputMultiset}{StabilizingInputMultiset}
\SetKwFunction{AlgStartInputMultisetTrial}{StartInputMultisetTrial}
\SetKwFunction{AlgTerminatingInputMultiset}{TerminatingInputMultiset}
\SetKwFunction{AlgStartTerminatingInputMultisetTrial}{StartTerminatingInputMultisetTrial}

\newcommand{\boolname}[1]{\ifmmode\text{\normalfont\textsc{#1}}\else\textnormal{\textsc{#1}}\fi}

\title{Computing in Anonymous Dynamic Networks with One-Bit Communications}

\author{
Thibaut Blanc (ENS de Lyon)\\
\texttt{thibaut.blanc@ens-lyon.fr}
\and
Giuseppe Antonio Di Luna (Sapienza University of Rome)\\
\texttt{diluna@diag.uniroma1.it}
\and
Giovanni Viglietta (University of Aizu)\\
\texttt{viglietta@gmail.com}
}
\date{}

\begin{document}

\maketitle
\begin{abstract}
We initiate the study of deterministic computation in anonymous dynamic networks in which each agent broadcasts a single bit per round and receives only the number of neighbors that broadcast each bit value. Despite this minimal communication, we show that surprisingly rich global computation remains possible.

When the network has a unique leader and an upper bound \(U\) on the network size \(n\) is known, we give a terminating algorithm for any desired computable function of the input multiset in \(O(n^3\log^2 n+U)\) rounds, where inputs are drawn from a universe of size \(N=2^{O(n\log n)}\). In addition, without any prior knowledge of \(n\), we design a stabilizing algorithm for the same task that runs in \(O(n^3\log^2 n)\) rounds. Notably, this essentially matches the state of the art for the congested communication model, where messages may carry \(O(\log n)\) bits rather than just one, and general computation is achieved in \(O(n^3)\) rounds. We also obtain companion results for leaderless and multi-leader networks, with comparable performance.

We complement these upper bounds with an almost-matching lower bound of
\[
\Omega\!\left(\frac{n^2\log(N/n)}{\log n}\right)
\]
rounds, which becomes \(\Omega(n^3)\) when \(N=2^{\Omega(n\log n)}\). The proof is an information-theoretic argument on local histories, and the lower bound holds even when the network has a unique leader, \(n\) and \(N\) are known, and the communication graph is restricted to a ring that may change every round.

Our algorithmic techniques are based on extracting global linear equations from local one-bit aggregate observations. A one-bit cut test gives a conservation constraint on the sizes of indistinguishable agent classes; by repeatedly refining these classes and collecting independent constraints, the agents recover the desired multiplicities. For the unknown-size case, we also introduce a self-correcting adaptive flooding primitive of independent interest. Together, these results show that the computational power of congested anonymous dynamic networks is essentially preserved, even when every message is compressed to a single bit.
\end{abstract}

\section{Introduction}
\label{sec:introduction}

Dynamic networks are a central research topic in distributed computing.
In these systems, communication links may change unpredictably over time.
One of the most widely studied models is the
\emph{\(1\)-interval-connected} model. In this synchronous model, a system of $n$ agents
proceeds in rounds and, in every round \(t\), an adversary selects
an arbitrary communication graph \(G_t\) subject only to the requirement
that \(G_t\) be connected.

In this paper, we focus on \emph{anonymous} dynamic networks. Agents
do not have unique identifiers and execute the same deterministic
algorithm. Except for their input, agents start in the same state. Anonymous networks arise naturally when identifiers are unavailable, undesirable, or incompatible
with the application, for example in privacy-sensitive systems,
large-scale populations of simple devices, and biological systems in
which globally unique identifiers do not exist. From an algorithmic
perspective, however, anonymity creates a fundamental symmetry problem:
agents with identical local histories must remain in identical states.

It is often customary to make the minimal symmetry-breaking assumption that the network contains a set of \emph{leader} agents, namely a set of agents with a distinguished initial state whose exact size $k$ is known. The assumption of having a unique leader is the special case in which $k=1$. A symmetry-breaking assumption is generally necessary \cite{michail_chatzigiannakis_spirakis_SSS2013} for distributed algorithms computing functions that depend on the absolute scale of the system, such as \ProblemName{Counting}, which returns the exact number of agents $n$, and \ProblemName{Input Multiset}, which returns the multiset of initial inputs of the agents.

If broadcasting messages of unbounded size is allowed, the presence of a unique leader
is sufficient to reconstruct the complete multiset of inputs. More
generally, every function that is deterministically computable in anonymous
\(1\)-interval-connected networks can be computed in a linear number of
rounds, thanks to a data structure called a \emph{history tree} \cite{diluna_viglietta_FOCS2022}.

A more restrictive setting is the \emph{congested} model, in which each
message contains at most \(O(\log n)\) bits. Despite this bandwidth
restriction, when a unique leader is present, arbitrary functions of the input multiset can still be computed in \(O(n^3)\) rounds \cite{diluna_viglietta_DC2025_congested}. This raises a natural question: how much communication bandwidth is actually required for general computation in anonymous dynamic networks?

We study a minimal transmission model: in every round, each agent emits a single bit. The chosen bit is broadcast to all of its current neighbors. An agent learns the exact number of its neighbors that sent $0$ and the exact number that sent $1$ in the current round.\footnote{In the related beeping model, in which an agent either beeps or listens and cannot count the multiplicity of received beeps, computing the \ProblemName{Counting} function is impossible in an anonymous network with a leader. For example, consider a complete graph with a leader.}

\subsection{Contributions}
\label{subsec:contributions}

We consider \emph{stabilizing} and \emph{terminating} algorithms. An algorithm \emph{stabilizes} if, after a finite number of rounds, the outputs returned by all agents are correct and never change again; their internal states may still keep evolving. An algorithm \emph{terminates} if it not only stabilizes, but every agent eventually enters a final state, which no longer changes.

\paragraph{A quadratic lower bound.}
 We prove that computing the
\ProblemName{Input Set}, that is the set of initial inputs of agents without multiplicity, where inputs are drawn from a universe of size
\(N\geq 2n\), requires
\[
    \Omega\!\left(
        \frac{n^2\log(N/n)}{\log n}
    \right)
\]
rounds. This lower bound holds even when all agents have distinct input values (which is equivalent to having unique identifiers), all agents know \(n\) and \(N\), the network contains a unique leader, the communication graph is a ring in every round, and the algorithm is required only to stabilize. In particular,
when \(N=n^{1+\varepsilon}\) for any constant
\(\varepsilon>0\), the bound becomes $\Omega(n^2)$; when \(N=2^{n\log n}\), the bound becomes \(\Omega(n^3)\).

Note that the same lower bound holds for the \ProblemName{Input Multiset} problem, which also requires determining the number of agents that have been assigned each input value.

We then prove that this lower bound is nearly optimal by designing a terminating algorithm for \ProblemName{Input Set}, whose running time matches the lower bound up to logarithmic factors, provided that a linear upper bound on $n$ is known by the agents.

\paragraph{Leaderless \ProblemName{Input Frequency} and \ProblemName{Input Multiset} in the multi-leader setting.}
One of our main algorithmic contributions is a method for converting local counts into linear equations describing the multiplicities of classes of indistinguishable agents. The algorithm repeatedly refines these classes and collects independent equations until their relative frequencies are uniquely determined.

Given a common upper bound \(U\geq n\), the algorithm solves \ProblemName{Input Frequency}, the problem of computing the relative frequency of each input, without a leader in $O\!\left(UnB_{\max}+Un^2\log n \right)$ rounds, where \(B_{\max}\) is the maximum bit length of an input value.

We remark that computing input frequencies is complete for the class of problems solvable in leaderless anonymous dynamic networks \cite{diluna_viglietta_DISC2023_disconnected}, and it yields an immediate solution to the well-studied average consensus problem \cite{olshevsky_SICON2017}.

When a known set of leaders is available, that is, when the class containing the leaders has known positive multiplicity,\footnote{This is also the case in which one distinguished input value and its multiplicity are known a priori.} we can use the number of leaders as a fixed normalization factor for the linear system.
This allows every agent to recover the multiplicity of every input value. We therefore obtain algorithms for \ProblemName{Input Multiset} and \ProblemName{Counting} with the same asymptotic round complexity. For \(U=\Theta(n)\) and \(B_{\max}=O(\log n)\), the resulting complexity is $O(n^3\log n)=\widetilde{O}(n^3)$.

Thus, at the level of round complexity, restricting every sender to one-bit transmissions introduces only a logarithmic factor compared with the \(O(n^3)\)-round algorithm known for congested anonymous dynamic networks \cite{diluna_viglietta_DC2025_congested}.

Under the mild additional assumption that, at the beginning of each round, every agent can query an oracle returning its current degree, we show how a known number of leaders can be used to compute an exponential upper bound \(U\) on the network size.
Combining this bound with our known-bound procedure described above leads to a terminating, albeit exponential-time, algorithm that requires no prior global knowledge of the network.

\paragraph{A stabilizing polynomial algorithm with a single leader and no prior knowledge.}
Finally, in the case of a unique leader, we remove the assumption that agents initially know a common upper bound on the network size, as well as any other a priori assumption or oracle. We introduce a technique for building a stabilizing flooding algorithm that broadcasts a message using a self-correcting network-size estimate. The broadcast either completes or enters a special error state that ends up modifying the estimate. Based on this primitive, we obtain a stabilizing algorithm that does not explicitly terminate but computes \ProblemName{Input Multiset}, and thus \ProblemName{Counting}, with round complexity $O\!\left( n^2B_{\max}\log n+n^3\log^2 n\right)$
without any initial knowledge of \(n\).

If \(U\geq n\) is given, this technique yields a terminating algorithm for \ProblemName{Input Multiset} with complexity $O\!\left(n^2 B_{\max}\log n + n^3 \log^2 n + U\right)$. Note that this algorithm achieves a slightly better running time than the one in the previous section, although it requires a unique leader.

Interestingly, if the input universe consists of the integers from \(1\) to \(N=2^{n\log n}\), then we have \(B_{\max}=n\log n\), and the previous upper bounds reduce to \(O(n^3\log^2 n)\), matching our lower bound of \(\Omega(n^3)\) up to logarithmic factors. For larger input universes, we have again almost-matching upper and lower bounds of \(O(n^2\log N \log n)\) and \(\Omega(n^2\log N/\log n)\), respectively.

Table~\ref{tab:results} summarizes the main results of this paper.


\begin{table}[ht]
\centering
\scriptsize
\setlength{\tabcolsep}{4pt}
\renewcommand{\arraystretch}{1.18}

\begin{tabularx}{\linewidth}{@{}l l X@{}}
\toprule
\textbf{Problem} &
\textbf{Assumptions} &
\textbf{Result} \\
\midrule

\makecell[l]{\ProblemName{Input Set}\\\ProblemName{Input Multiset}} &
\makecell[l]{Unique leader\\\(N\geq 2n\)} &
Stabilization requires
\(\Omega\!\left(n^2\log(N/n)/\log n\right)\) rounds
(\S~\ref{sec:lower-bound-bc}) \\
\arrayrulecolor{black!15}\specialrule{0.35pt}{0pt}{0pt}\arrayrulecolor{black}

\ProblemName{Input Set} &
\makecell[l]{Known \(U\)} &
Terminating algorithm in $O(Un(1+B_{\max}))$ rounds
(\S~\ref{sec:inputset}) \\
\arrayrulecolor{black!15}\specialrule{0.35pt}{0pt}{0pt}\arrayrulecolor{black}

\ProblemName{Input Frequency} &
\makecell[l]{Known \(U\)} &
Terminating algorithm in \(O(Un B_{\max}+Un^2\log n)\) rounds
(\S~\ref{sec:input-frequency}) \\
\arrayrulecolor{black!15}\specialrule{0.35pt}{0pt}{0pt}\arrayrulecolor{black}

\ProblemName{Input Multiset} &
\makecell[l]{Multi-leader\\Known \(U\)} &
Terminating algorithm in \(O(Un B_{\max}+Un^2\log n)\) rounds
(\S~\ref{sec:frequency-to-multiset}) \\
\arrayrulecolor{black!15}\specialrule{0.35pt}{0pt}{0pt}\arrayrulecolor{black}

\ProblemName{Counting} &
\makecell[l]{Multi-leader\\Known \(U\)} &
Terminating algorithm in \(O(Un^2\log n)\) rounds
(\S~\ref{sec:frequency-to-multiset}) \\
\arrayrulecolor{black!15}\specialrule{0.35pt}{0pt}{0pt}\arrayrulecolor{black}

\makecell[l]{\ProblemName{Input Frequency}\\\ProblemName{Input Multiset}} &
\makecell[l]{Multi-leader\\Local degree oracle} &
Terminating algorithm in \(O(n^nB_{\max}+n^{n+1}\log n)\) rounds
(\S~\ref{sec:upperbound}) \\
\arrayrulecolor{black!15}\specialrule{0.35pt}{0pt}{0pt}\arrayrulecolor{black}

\ProblemName{Counting} &
\makecell[l]{Multi-leader\\Local degree oracle} &
Terminating algorithm in \(O(n^{n+1}\log n)\) rounds
(\S~\ref{sec:upperbound}) \\
\arrayrulecolor{black!15}\specialrule{0.35pt}{0pt}{0pt}\arrayrulecolor{black}

\ProblemName{Input Multiset} &
Unique leader &
Stabilizing algorithm in
\(O(n^2 B_{\max}\log n+n^3\log^2 n)\) rounds
(\S~\ref{sec:stabilizing-input-multiset-counting}) \\
\arrayrulecolor{black!15}\specialrule{0.35pt}{0pt}{0pt}\arrayrulecolor{black}

\ProblemName{Counting} &
Unique leader &
Stabilizing algorithm in \(O(n^3\log^2 n)\) rounds
(\S~\ref{sec:stabilizing-input-multiset-counting}) \\
\arrayrulecolor{black!15}\specialrule{0.35pt}{0pt}{0pt}\arrayrulecolor{black}

\ProblemName{Input Multiset} &
\makecell[l]{Unique leader\\Known \(U\)} &
Terminating algorithm in
\(O(n^2 B_{\max}\log n+n^3\log^2 n+U)\) rounds
(\S~\ref{subsec:bc-terminating-common-bound}) \\
\arrayrulecolor{black!15}\specialrule{0.35pt}{0pt}{0pt}\arrayrulecolor{black}

\ProblemName{Counting} &
\makecell[l]{Unique leader\\Known \(U\)} &
Terminating algorithm in \(O(n^3\log^2 n+U)\) rounds
(\S~\ref{subsec:bc-terminating-common-bound}) \\

\bottomrule
\end{tabularx}

\caption{%
Summary of the main results. Here, \(B_{\max}\) denotes the maximum bit length of an input value, \(U\geq n\) denotes an upper bound on the network size, and \(N\)
denotes the input universe size. For the multi-leader
assumption we have a set of \(k\geq 1\) leaders with \(k\) known; the unique-leader assumption is the special case where \(k=1\).
}
\label{tab:results}
\end{table}

\paragraph{Challenges and significance of the results.} These results show that general computation in anonymous dynamic networks
remains possible even when communication is reduced below the congested
threshold. To the best of our knowledge, this is the first work to study
deterministic global computation in anonymous dynamic networks with one-bit
broadcast communication. The setting is particularly hostile because
three restrictions interact: agents have no identifiers, the topology changes
adversarially from round to round, and messages contain only one bit. The
difficulty is not merely that messages are short. In an anonymous dynamic
network, a long message cannot simply be divided into small pieces and
reassembled later, because the pieces carry no persistent source information and
the neighbors of an agent may change before the next piece is received. In the
congested model, this obstacle can still be overcome by sending \(O(\log n)\)-bit
objects, such as temporary labels or edge descriptors \cite{diluna_viglietta_DC2025_congested}. In our model, even this
is unavailable: an agent broadcasts only a bit, and its feedback is only the
number of neighbors that sent each bit. Thus agents cannot transmit temporary
identifiers, attach names to fragments of information, or send numerical
quantities directly.

This makes the comparison with history-tree algorithms especially informative.
History trees were a breakthrough in the study of anonymous dynamic networks:
they led to optimal linear-time computability in the non-congested model and
were later adapted to obtain the \(O(n^3)\)-round state of the art for congested
anonymous dynamic networks
\cite{diluna_viglietta_FOCS2022,diluna_viglietta_DISC2023_disconnected,diluna_viglietta_DC2025_congested}.
The congested adaptation still builds on the history-tree paradigm: it
constructs and transmits compacted history-tree information using temporary
non-unique identifiers and carefully controlled broadcasts. Our algorithms do
not construct history trees, do not exchange vistas, and do not simulate the
congested history-tree transmission. Nevertheless, for a linear upper bound
\(U=\Theta(n)\) and \(B_{\max}=O(\log n)\), our terminating algorithms compute
input multiplicities in \(O(n^3\log n)\) rounds, and our unique-leader
stabilizing algorithm without prior knowledge computes them in
\(O(n^3\log^2 n)\) rounds. Hence one-bit communication matches the congested
state of the art up to logarithmic factors.

Our unique-leader algorithm also introduces a self-correcting broadcast primitive
that is of independent interest. The adaptive flooding layer of
Section~\ref{sec:anonymous-bc-unknown-bound} replaces the usual assumption of a
known upper bound \(U\geq n\) by a speculative estimate \(\widehat U\). If
\(\widehat U\geq n\), the primitive behaves like an ordinary flooding operation; otherwise, the certificates expose the inconsistency and trigger a restart with a larger estimate. Thus the layer provides a reusable way to convert flooding-based algorithms that assume a known size bound into stabilizing algorithms that require no prior knowledge of the
network size.

The results in this paper also separate our approach from the older local averaging and
mass-distribution methods used for counting and average-consensus-type tasks in
anonymous dynamic networks
\cite{diluna_baldoni_bonomi_chatzigiannakis_ICDCS2014,kowalski_mosteiro_ICALP2018,kowalski_mosteiro_JCSS2021,olshevsky_SICON2017}.
Those techniques rely on the ability to move and compare numerical quantities
through the network. Such operations have no direct analogue in the one-bit
broadcast-counting model. Instead, our algorithms extract information from the
only numerical data that remain available: the aggregate counts of neighbors
that chose \(0\) and \(1\). A one-bit cut test gives a linear conservation
constraint on the sizes of refined indistinguishability classes; after enough
independent constraints have been collected, the class frequencies, and then the
input multiplicities when leaders are available, are determined.

Thus the paper identifies a different route to global computation: not by reconstructing
history, transmitting identifiers, or redistributing mass, but by turning
one-bit aggregate observations into a solvable global linear system. Furthermore, for
logarithmic-size inputs, this representation uses only \(O(n^2\log n)\) bits of
local state. This is substantially smaller than storing
history-tree vistas which, after $t$ rounds, require \(O(tn^2\log n)\) bits of local state in the worst case \cite{viglietta_SIROCCO2024}.

\paragraph{Paper overview}
Section~\ref{sec:sysmod} defines the model and tasks.
Section~\ref{sec:lower-bound-bc} proves the input-set lower bound.
Section~\ref{sec:basic} gives the basic one-bit primitives and the near-optimal input-set algorithm.
Section~\ref{sec:leader-anonymous-bc} gives known-bound algorithms for \ProblemName{Input Frequency} and \ProblemName{Input Multiset}.
Section~\ref{sec:anonymous-bc-unknown-bound} removes the known-bound assumption in the unique-leader case by introducing the self-correcting adaptive flooding layer.
Section~\ref{sec:stabilizing-input-multiset-counting} uses this layer to build the stabilizing \ProblemName{Input Multiset} algorithm.
Section~\ref{sec:related-work} discusses related work.

\section{System Model} \label{sec:sysmod}
\paragraph{Agents and interval connectivity.}
The system consists of a finite, nonempty set $V$ of $n$ agents.
Time is divided into synchronous rounds. In round $t\geq 1$, the
communication topology is a simple undirected graph $G_t=(V,E_t)$ chosen by an adversary. A dynamic network is \emph{\(1\)-interval-connected} if $G_t$ is connected
in every round $t$.

\paragraph{Computation, anonymity and leaders.}
In an anonymous system, agents have no identifiers. Each agent $v$ starts with an input $\lambda(v)$, which is assigned to it at round $0$. When relevant, its local data also include its leader flag, and its initial state may depend on the common knowledge assumed, such as knowledge of a \(U \geq n\).
 When we discuss local space, we mean the number of bits needed
to encode an agent's local state that includes its output register.

At the beginning of each round $t \geq 0$, during the communication step, each agent chooses, based on its internal state, whether to send $0$ or $1$. This message is then delivered to all its neighbors.  After the communication step and before the end of the round, the agent learns the number of neighbors that sent $0$ and the number of neighbors that sent $1$. At this point, each agent uses its state and the multiset of received messages to update its state according to a deterministic algorithm ${A}$, and then proceeds to the next round.\footnote{An agent receives only a multiset of $0$'s and $1$'s, and no other information, such as port numbers or similar identifiers.} Note that ${A}$ is the same for all agents. We call this the \emph{one-bit broadcast-counting model}.

When we assume the presence of multiple leaders, we mean that the input of each agent also includes a distinguished input flag, called the leader flag, and that exactly $k$ agents start with this flag on, where $k$ is known to all agents. The unique-leader setting is the special case $k=1$.

When we assume an upper bound on the network size, we assume that there is a quantity $U\geq n$ known to all agents.

In one result, we also assume the presence of a \emph{local degree oracle}. In this case, each agent $v$ has access to an oracle that it can query at the beginning of each round $t$ to learn the number of its neighbors in $G_t$.

\paragraph{Stabilization and termination.}
Each agent produces an output at the end of every round. An execution
\emph{stabilizes} if, from some round onward, every agent has the
correct output and no output changes. Internal states may continue to
evolve after stabilization.

An agent \emph{terminates} when it enters a final state from which no
further transition is possible. An execution terminates when every
agent has terminated.

The \emph{stabilization time} is the least $T$ such that, at the end of every round $t\geq T$, every agent has the correct output and no output changes thereafter. The \emph{termination time} of an execution is the round in which all agents have terminated.
An \emph{admissible execution} is an execution consistent with the algorithm, the assumed inputs and leader flags, the assumed initial knowledge, and an adversarial sequence of \(1\)-interval-connected communication graphs.
An algorithm \emph{computes} a function if every admissible execution
stabilizes to the correct output. It computes the function \emph{with
termination} if every agent also eventually terminates.

\paragraph{Functions and Problems.}

An \emph{input assignment} is a map \(\lambda:V\to\mathbb N_{>0}\), together with any leader flags required by the theorem under consideration. The global input multiset of \(\lambda\) is
\[
    \mu_\lambda=\{\!\{\,\lambda(v):v\in V\,\}\!\}.
\]

Unless a theorem specifies a finite input universe, input values are positive integers. We write
$B_{\max}= \max_{v\in V} \left\lceil \log(\lambda(v)+1) \right\rceil$
for the maximum input bit length.
We investigate the following problems:
\begin{itemize}

\item \textbf{\ProblemName{Counting}.}
The \ProblemName{Counting} function is $F_C(v,\lambda)=n$: every agent must
output the size of the system.

\item \textbf{\ProblemName{Input Multiset}.}
The \ProblemName{Input Multiset} function is
$F_{IM}(v,\lambda)=\mu_\lambda$. Equivalently, for each input value $a$,
every agent must output its multiplicity $|\{u\in V:\lambda(u)=a\}|$. \ProblemName{Counting} reduces to \ProblemName{Input Multiset} by having every agent ignore its original
input and invoke the multiset routine with the same synthetic constant input.

\item \textbf{\ProblemName{Input Set}.}
The \ProblemName{Input Set} function is $F_{IS}(v,\lambda)=\{\lambda(u):u\in V\}$. Equivalently, every agent must output the set of input values that occur in the network, ignoring their multiplicities.

\item \textbf{\ProblemName{Input Frequency}.}
The \ProblemName{Input Frequency} function maps each input value \(a\) to its
relative multiplicity
\[
    \frac{|\{u\in V:\lambda(u)=a\}|}{n}.
\]
Equivalently, it is the normalized version \(\frac{1}{n}\mu_\lambda\) of the
input multiset.

\item \textbf{\ProblemName{Counting} upper bound.}
Every agent must eventually output the same finite integer $U$ satisfying \(U\geq n\).
\end{itemize}

\section{Lower Bound for Input-Set Computation}
\label{sec:lower-bound-bc}

We prove that computing the \ProblemName{Input Set} function requires
\[
    \Omega\!\left(\frac{n^2\log(N/n)}{\log n}\right)
\]
rounds when the input values are drawn from a universe of size $N$. In particular, if $N=n^{1+\varepsilon}$ for a fixed $\varepsilon>0$, the lower bound is $\Omega(n^2)$, even though each input has only $O(\log n)$ bits. This remains true even if the network has a
unique leader and even if all agents know both $n$ and $N$. Clearly, the same lower bound applies to the \ProblemName{Input Multiset} problem.

The proof is a counting argument on local histories. The adversary chooses each communication graph so that, in every round, all but a constant number of agents receive an observation that is determined by their current local history. We will show that, if too few rounds have passed, some agent has too few possible local histories to distinguish all possible input sets. Hence two different input assignments remain indistinguishable to that agent.

\begin{theorem}
\label{thm:bc-input-set-lower-bound-general-recall}
Suppose that input values are drawn from a universe of size $N\geq 2n$. In anonymous \(1\)-interval-connected dynamic networks, every deterministic algorithm that computes the Input
Set function requires
\[
    \Omega\!\left(
        \frac{n^2\log(N/n)}{\log n}
    \right)
\]
rounds in the worst case. This holds even if there is a unique leader, even if all
agents know $n$ and $N$, even if all agents are assigned different input values, and
even if the communication graph at every round is a ring (not necessarily the same ring in every round).
\end{theorem}

\begin{proof}
Fix a deterministic algorithm $A$. We prove the lower bound in the
stronger setting in which there is a unique leader and all agents
know both $n$ and $N$, and restrict the adversary to choosing a ring in
every round. Giving information to the agents and restricting the
adversary can only help the algorithm.
Give the unique leader a fixed input value that is
not assigned to any other agent. The remaining $n-1$ non-leader agents
receive distinct input values from a set $U_N$ of size $N$, assigned to them in a well-defined canonical order. We consider
one execution for every set
\[
    S\subseteq U_N,\qquad |S|=n-1.
\]

The adversary will choose a ring as the communication graph in every
round. Hence every agent always has degree $d=2$, and its observation is completely determined by the number
$q\in\{0,1,2\}$ of its neighbors that sent $1$. We therefore identify an
observation with $q$ throughout the proof.

We now describe the adversary. Consider the beginning of a round, before
the graph of the round has been chosen. For each agent $v$, let
$b_v$ be the bit that $v$ sends in this round, given its current local
history and degree $2$. Since the algorithm $A$ is
deterministic and the adversary knows the execution so far, the values
$b_v$ are well defined. Let
\[
    V^0=\{v:b_v=0\},
    \qquad
    V^1=\{v:b_v=1\}.
\]

The adversary arranges the agents on a ring so that the agents in
$V^0$ form one contiguous block and the agents in $V^1$ form the
other. If one of the two classes is empty, it chooses an arbitrary ring.
Thus, the communication graph is a ring in every round, although its
cyclic order may change between rounds.

We call an agent $v\in V^b$ \emph{active} in this round if at least one
of its neighbors belongs to $V^{1-b}$. All other agents are called
\emph{inactive} in this round. Since a ring
containing two nonempty contiguous blocks has exactly two edges between
the blocks, at most $C=4$ agents are active in each round.

Observe that every inactive agent in a class $V^b$ has degree $2$, and both of its neighbors
belong to the same class $V^b$. Therefore, each of these three agents sends a bit according to the prediction used to
define the classes $V^0$ and $V^1$. It follows that an inactive agent
observes a number of ones that is:
\[
    0 \quad\text{if } b=0,
    \qquad
    2 \quad\text{if } b=1.
\]

An active agent can have $D=3$ observations $0, 1$ and $2$.

For an execution and an agent $v$, let $H_v(T)$ be the local history of
$v$ up to time $T$: it consists of the initial state of $v$, including
its input value and whether it is the leader, together with the sequence
of its observations during the first $T$ rounds. Since the
algorithm $A$ is deterministic, $H_v(T)$ determines the state of $v$ at time
$T$.

Suppose that a non-leader agent $v$ is active in at most $R$ of the
first $T$ rounds. Then the number of possible histories $H_v(T)$ is at
most
\[
    N\sum_{r=0}^{R}\binom{T}{r}D^r .
\]
Indeed, such a history is determined by:
\begin{enumerate}
\item the input value of $v$, with at most $N$ choices;
\item the set of active rounds, with $\binom{T}{r}$ choices if there are
$r$ active rounds;
\item the observations in the active rounds, with at most $D^r$ choices.
\end{enumerate}
The inactive observations do not contribute to this count, because they are forced. Given the
history of $v$ before a inactive round, the algorithm determines the
bit $b$ that $v$ sends. Since the round is
inactive, the observation is then necessarily
\[
    0 \quad\text{if } b=0,
    \qquad
    2 \quad\text{if } b=1.
\]

Since at most $C$ agents are active in any one round, during the
first $T$ rounds there are at most $CT$ pairs $(v,t)$ such that agent
$v$ is active in round $t$. At least half of the $n-1$ non-leader agents must be active in at most
\[
    R=\left\lceil \frac{2CT}{n-1}\right\rceil
\]
rounds. Indeed, if more than half of the non-leader agents were
active in more than $R$ rounds, then these agents alone would account
for more than $CT$ such pairs, a contradiction.

We call a pair $(S,v)$ \emph{quiet} if $S\subseteq U_N$ is one of the
input sets under consideration, $v$ is a non-leader agent in the
execution associated with $S$, and $v$ is active in at most $R$ of the
first $T$ rounds. Hence, for every input set $S$, at least $(n-1)/2$
choices of $v$ give a quiet pair. Therefore the total number of quiet pairs is
at least
\[
    \frac{n-1}{2}\binom{N}{n-1}.
\]
On the other hand, as shown above, the number of possible local
histories of an agent that is active in at most \(R\) rounds is at most
\[
    N\sum_{r=0}^{R}\binom{T}{r}D^r .
\]
Thus this also bounds the number of possible local histories that can
occur in quiet pairs.

We compare these two quantities. Since $N\geq 2n$,
\[
    \binom{N}{n-1}
    \geq
    \left(\frac{N}{n-1}\right)^{n-1}
    \geq
    \left(\frac{N}{n}\right)^{n-1}.
\]
Thus
\[
    \log\left(
        \frac{n-1}{2}\binom{N}{n-1}
    \right)
    \geq
    (n-1)\log(N/n)-O(\log n).
\]

Now set
\[
    T=
    c\,\frac{n^2\log(N/n)}{\log n}
\]
for a constant $c>0$ to be chosen sufficiently small. Then
\[
    R=\left\lceil \frac{2CT}{n-1}\right\rceil
    =
    O\!\left(
        c\,\frac{n\log(N/n)}{\log n}
    \right).
\]

We now count the number of possible local histories that can arise from
a quiet pair. From Stirling's formula, we have $r!\geq (r/e)^r$. Hence, for every $1\leq r\leq R$,
\[
    \binom{T}{r}
    =
    \frac{T(T-1)\cdots(T-r+1)}{r!}
    \leq
    \frac{T^r}{r!}
    \leq
    \left(\frac{eT}{r}\right)^r .
\]
Therefore
\[
    \binom{T}{r}D^r
    \leq
    \left(\frac{eDT}{r}\right)^r .
\]

The function
\[
    x\mapsto x\log\left(\frac{eDT}{x}\right)
\]
is increasing for $x\leq DT$, because its derivative is
\[
    \log(DT/x).
\]
For all sufficiently large $n$, we have $R\leq T\leq DT$. Hence every
term with $1\leq r\leq R$ is at most
\[
    \left(\frac{eDT}{R}\right)^R .
\]
The term $r=0$ is $1$, which is also bounded by the same expression.
Thus
\[
    \sum_{r=0}^{R}\binom{T}{r}D^r
    \leq
    (R+1)\left(\frac{eDT}{R}\right)^R .
\]

Since
\[
    R=\left\lceil \frac{2CT}{n-1}\right\rceil,
\]
we have $T/R=O(n)$. Therefore
\[
\begin{aligned}
    \log\left(
        N\sum_{r=0}^{R}\binom{T}{r}D^r
    \right)
    &\leq
    \log N
    +
    \log(R+1)
    +
    R\log\left(\frac{eDT}{R}\right)        \\[1mm]
    &=
    \log N
    +
    \log(R+1)
    +
    R\left(\log(eD)+\log\left(\frac{T}{R}\right)\right) \\[1mm]
    &\leq
    \log N
    +
    R(\log n+O(1))
    +
    O(\log R)                              \\[1mm]
    &\leq
    \log N
    +
    O\!\left(c\,n\log(N/n)\right)
    +
    O(\log R),
\end{aligned}
\]
where we used our previous bound on $R$.

Set $L=\log(N/n)\geq 1$. Observe that both \(L\) and \(O(\log n)\) are \(o(nL)\). Hence the lower bound on the number of quiet pairs is
\[
    \log\left(
        \frac{n-1}{2}\binom{N}{n-1}
    \right)
    \geq
    (n-1)L-O(\log n)
    =
    nL-L-O(\log n)
    =
    nL-o(nL).
\]

The upper bound on the number of quiet local histories is
\[
    \log\left(
        N\sum_{r=0}^{R}\binom{T}{r}D^r
    \right)
    \leq
    \log N+O(cnL)+O(\log R)
    =
    O(cnL)+o(nL),
\]
because \(\log N=\log n+L=o(nL)\), and from
\[
    R=O\!\left(c\,\frac{nL}{\log n}\right)
\]
we also have \(\log R=o(nL)\).

Choosing \(c>0\) sufficiently small, the upper bound is strictly smaller
than the lower bound for all sufficiently large \(n\). Therefore
\[
    N\sum_{r=0}^{R}\binom{T}{r}D^r
    <
    \frac{n-1}{2}\binom{N}{n-1}.
\]
We conclude that there are more quiet pairs than possible local histories
that can arise from quiet pairs. By the pigeonhole principle, there exist two distinct quiet pairs $(S,v)$ and $(S',v')$ such that the local history of $v$ in the
execution associated with $S$ is equal to the local history of $v'$ in
the execution associated with $S'$, up to time $T$.

A local history contains the agent's input value; hence \(v\) and
\(v'\) have the same input value. We claim that \(S\neq S'\). Indeed, if
\(S=S'\), then, since all inputs in \(S\) are distinct and assigned to agents in a well-defined canonical order, there is only one
agent with that input value. Thus \(v=v'\), contradicting the fact
that $(S,v)$ and $(S',v')$ are distinct. Therefore \(S\neq S'\).

However, $v$ and $v'$ have identical local histories up to time $T$,
and hence identical states at time $T$. Since the algorithm $A$ is
deterministic, both agents produce the same output at time $T$. But the correct
\ProblemName{Input Set} outputs are different, because one execution has input set
$S$ together with the fixed leader input, while the other has input set
$S'$ together with the same fixed leader input. Thus the same output
cannot be correct in both executions.

Consequently, $A$ cannot guarantee correctness by time
\[
    c\,\frac{n^2\log(N/n)}{\log n}.
\]
This proves the claimed lower bound.
\end{proof}

As a complementary observation, the frequency problem also requires large local
state when the input universe is large.

\begin{remark}
\label{rem:bc-frequency-space-lower-bound}
Suppose that input values are drawn from a finite universe of size \(N\geq n\). Every deterministic algorithm that computes the \ProblemName{Input Frequency} function on networks of \(n\) agents requires at least $\Omega(n \log(1+\frac{N}{n}))$ bits of local space in the worst case. This holds even if all agents know \(n\) and \(N\), and even if the communication graph is static and complete in every round.
\end{remark}

\begin{proof}
Fix an agent \(v\) and fix its input to some value \(a\) in the input universe. Consider all possible multiplicity vectors for the inputs of the remaining \(n-1\) agents. Such a vector consists of \(N\) nonnegative integers whose sum is \(n-1\), so there are \(\binom{n+N-2}{N-1}\) possible vectors.

Each of these vectors gives a different correct \ProblemName{Input Frequency} output at \(v\). Indeed, because \(n\) is fixed, the frequency of each value uniquely determines its multiplicity. Thus \(v\) must be able to produce at least \(\binom{n+N-2}{N-1}\) distinct outputs while its own input remains fixed.

If \(v\) uses at most \(s\) bits of local space, it has at most \(2^s\) possible local states. Since the output of an agent is determined by its local state, \(v\) can produce at most \(2^s\) distinct outputs. Consequently, \(2^s\geq\binom{n+N-2}{N-1}\).
Now by the symmetry of the binomial coefficient we have $\binom{n+N-2}{N-1}=\binom{n+N-2}{n-1}$. Using the standard estimate $\binom{m}{k}\geq (m/k)^k$ up to constant factors in the exponent, we get $\log \binom{n+N-2}{n-1}=\Omega((n-1)\log(\frac{n+N-2}{n-1}))$. Since \(N\geq n\), this simplifies to $s \geq \Omega(n \log(1+\frac{N}{n}))$.
\end{proof}

From the above, when \(N=n^{1+\varepsilon}\) for some \(\varepsilon>0\), each agent
requires \(\Omega(n\log n)\) bits of local space in the worst case.

\section{Basic Primitives}
\label{sec:basic}
In this section, we discuss some basic primitives that will be used by our algorithms and the pseudocode convention that we will use. We first present the \Flood algorithm, a primitive that computes the collective OR of the bits held by all agents. On top of this primitive, we build another primitive that computes the \ProblemName{Input Set} of the agents' values. Recall that this set contains the input values, but not their multiplicities. Both procedures require an upper bound $U$ on the network size.

\subsection{Computing the OR: \Flood}

The \Flood procedure
takes a duration $k$ and a local boolean value $b$. In each of the next
$k$ rounds, every agent sends its current value of $b$ to its
neighbors. If an agent receives at least one value equal to $1$, it
sets $b$ to $1$. After the $k$ communication rounds, every agent
returns $b$. We will show that, if \(k\geq n-1\), then, at the end of the procedure, each
agent returns the correct OR of all initial boolean values.
The pseudocode is given in Algorithm~\ref{alg:flood}.

\paragraph{Pseudocode convention.}
We also use Algorithm~\ref{alg:flood} to explain our pseudocode convention, which follows a standard convention of synchronous distributed algorithms (see, for example, \cite{lynch1996distributed}).

All pseudocode is written for an arbitrary agent, and the same code is
executed by every agent. Unless stated otherwise, all variables are local.
A statement ``send \(b\)'' represents one synchronous round: every agent
broadcasts the value of \(b\), computed from its local state, to its current
neighbors. After this statement, the agent knows how many neighbors sent
\(0\) and how many sent \(1\). All other statements are local computations
and do not consume rounds. Hence a loop with \(k\) iterations and one
``send'' statement per iteration uses \(k\) rounds.

\begin{algorithm}[H]
\caption{\Flood}
\label{alg:flood}

\KwInput{A duration $k\geq0$ and a boolean value $b \in \{0,1\}$}
\KwOutput{If $k \geq n-1$, returns the OR of all initial boolean values}

\Fn{\AlgFlood{$k,b$}}{
  \For{$t \leftarrow 1$ \KwTo $k$}{
    send $b$\;

    \If{at least one neighbor sent $1$}{
      $b \leftarrow 1$\;
    }
  }

  \Return{$b$}\;
}
\end{algorithm}

\begin{lemma}
\label{lem:flood-correctness}
Consider a \(1\)-interval-connected dynamic network. Let $S_0$ be the set
of agents whose initial boolean value is $1$. After \Flood$(k,b)$:
\begin{enumerate}
    \item an agent returns $1$ only if $S_0\neq\emptyset$;
    \item if $S_0\neq\emptyset$ and $|V\setminus S_0|\leq k$, then every
    agent returns $1$;
    \item the procedure uses exactly $k$ rounds.
\end{enumerate}
\end{lemma}

\begin{proof}
For $r\in\{0,\ldots,k\}$, let $S_r$ be the set of agents whose
boolean value is $1$ after the first $r$ rounds. The update is
irreversible, so $S_0\subseteq S_1\subseteq\cdots\subseteq S_k$.

If $S_0=\emptyset$, no agent sends $1$ in the first round.
Consequently, no agent changes its value and $S_1=\emptyset$.
Repeating the same argument inductively gives $S_r=\emptyset$ for every
$r\leq k$. Therefore no agent can return $1$ unless some agent
initially had value $1$.

Suppose now that $S_0\neq\emptyset$. Consider a round
$r\in\{1,\ldots,k\}$ such that $S_{r-1}\neq V$. Both
$S_{r-1}$ and $V\setminus S_{r-1}$ are nonempty. Since the graph $G_r$
of that round is connected, it contains an edge
$\{u,v\}\in E_r$ with $u\in S_{r-1}$ and $v\notin S_{r-1}$.

Agent $u$ sends $1$ during round $r$. Hence $v$ receives at least one
value equal to $1$ and sets its value to $1$. Thus $|S_r|\geq |S_{r-1}|+1$
whenever $S_{r-1}\neq V$. It follows by induction that after at most
$|V\setminus S_0|$ rounds all agents belong to the flooded set.
Therefore, if $|V\setminus S_0|\leq k$, then $S_k=V$ and every agent
returns $1$.

Finally, the procedure executes one round in each of its
$k$ iterations, so it uses exactly $k$ rounds.
\end{proof}

\begin{corollary}
\label{cor:flood-upper-bound}
If all agents know a common upper bound $U\geq n$, then
\Flood$(U-1,b)$ returns $\bigvee_{v\in V} b_v$
at every agent, where $b_v$ is the initial boolean value of agent
$v$.
\end{corollary}

\begin{proof}
If every initial value is $0$, the first part of
Lemma~\ref{lem:flood-correctness} implies that every agent returns
$0$. Otherwise $S_0\neq\emptyset$, and therefore $|V\setminus S_0|\leq n-1\leq U-1$.
The second part of the lemma implies that every agent returns $1$.
In both cases, the returned value is the OR of the initial values.
\end{proof}

\subsection{Computing the \ProblemName{Input Set}: \DistinctValues}\label{sec:inputset}

The \DistinctValues procedure is given in Algorithm~\ref{alg:distinct-values}.
Every agent $v$ invokes it with its local value $x=\lambda(v)$.

The first phase computes the bit length of the largest input value. The
counter $\ell$ is initialized in line~\ref{line:dv-init-length}, and the
loop in lines~\ref{line:dv-test-large}--\ref{line:dv-increase-length}
uses \Flood to find the smallest $\ell$ such that
$2^\ell>X_{\max}$.

The second phase extracts the distinct values in decreasing order. The
main loop starts in line~\ref{line:dv-remaining-loop}. In each iteration,
the currently remaining agents enter a bit-by-bit race, initialized
in line~\ref{line:dv-init-candidate}. The bit loop in
lines~\ref{line:dv-bit-loop}--\ref{line:dv-drop-candidate} keeps only
the candidates with the largest prefix seen so far. The value found in
this race is added to $I$ in line~\ref{line:dv-add-value}, and the
agents holding it are removed from later races in
line~\ref{line:dv-deactivate}. The final set is returned in
line~\ref{line:dv-return}.

\SetKw{KwDownTo}{downto}
\providecommand{\FloodCall}{\operatorname{\textnormal{\textsc{Flood}}}}

\begin{algorithm}[H]
\caption{\DistinctValues}
\label{alg:distinct-values}

\KwInput{A positive integer value $x$ and a common upper bound $U \geq n$}
\KwOutput{The set of input values initially present}

\Fn{\AlgDistinctValues{$x,U$}}{
  $\ell \leftarrow 0$\; \label{line:dv-init-length}

  \While{$\FloodCall(U-1,\, x \geq 2^\ell)$}{ \label{line:dv-test-large}
    $\ell \leftarrow \ell + 1$\; \label{line:dv-increase-length}
  }

  $I \leftarrow \emptyset$\; \label{line:dv-init-set}
  $\mathit{remaining} \leftarrow \boolname{true}$\; \label{line:dv-init-remaining}

  \While{$\FloodCall(U-1,\, \mathit{remaining})$}{ \label{line:dv-remaining-loop}
    $\mathit{candidate} \leftarrow \mathit{remaining}$\; \label{line:dv-init-candidate}
    $y \leftarrow 0$\; \label{line:dv-init-y}

    \For{$b \leftarrow \ell-1$ \KwDownTo $0$}{ \label{line:dv-bit-loop}
      $\mathit{hasOne} \leftarrow
        \FloodCall(U-1,\,
        \mathit{candidate} \land
        \text{bit } b \text{ of } x \text{ is } 1)$\; \label{line:dv-test-one}

      \If{$\mathit{hasOne}$}{
        set bit $b$ of $y$ to $1$\; \label{line:dv-set-bit}

        \If{$\mathit{candidate}$ and bit $b$ of $x$ is $0$}{
          $\mathit{candidate} \leftarrow \boolname{false}$\; \label{line:dv-drop-candidate}
        }
      }
    }

    $I \leftarrow I \cup \{y\}$\; \label{line:dv-add-value}

    \If{$\mathit{candidate}$}{
      $\mathit{remaining} \leftarrow \boolname{false}$\; \label{line:dv-deactivate}
    }
  }

  \Return{$I$}\; \label{line:dv-return}
}
\end{algorithm}

\begin{lemma}\label{lem:bc-distinct-values}
Let $B_{\max}$ be the bit length of the largest input value, and let
$q$ be the number of distinct input values. Assuming a common upper
bound $U\geq n$, \DistinctValues returns exactly the set of input values
initially present. It uses $O(Uq(1+B_{\max}))$ rounds, and
therefore $O(Un(1+B_{\max}))$ rounds.
\end{lemma}

\begin{proof}
Let $X_{\max}$ be the largest input value in the input assignment.
By definition of $B_{\max}$, all input values can be represented using
bits $B_{\max}-1,\ldots,0$, and $B_{\max}$ is the smallest integer such
that $2^{B_{\max}}>X_{\max}$.

First consider the loop of lines~\ref{line:dv-test-large}--%
\ref{line:dv-increase-length}. For a fixed value of $\ell$, the call to
\Flood in line~\ref{line:dv-test-large} returns true at every agent if
and only if some input value is at least $2^\ell$, by
Corollary~\ref{cor:flood-upper-bound}. Hence the loop increments
$\ell$ exactly while $2^\ell\leq X_{\max}$. When the loop terminates,
we have $\ell=B_{\max}$.

Now consider one iteration of the main loop in
line~\ref{line:dv-remaining-loop}, and let $W$ be the set of agents
with $\mathit{remaining}=\boolname{true}$ at the start of the iteration.
If $W=\emptyset$, the call to \Flood in line~\ref{line:dv-remaining-loop}
returns false, and the algorithm terminates. Otherwise,
line~\ref{line:dv-init-candidate} makes exactly the agents in $W$
candidates for the current race.

The loop starting in line~\ref{line:dv-bit-loop} scans the bits from
most significant to least significant. At bit position $b$,
line~\ref{line:dv-test-one} checks whether some current candidate has
bit $1$. If not, all candidates remain in the race. If so,
line~\ref{line:dv-set-bit} sets bit $b$ of $y$ to $1$, and
line~\ref{line:dv-drop-candidate} removes every candidate with bit $0$.
Thus, after each bit position, the remaining candidates are exactly
those agents in $W$ whose scanned prefix is maximum.

By induction over the bit positions, after the last bit the candidates
are exactly the agents in $W$ with maximum input value. The value
reconstructed in $y$ is therefore this maximum remaining value.
Line~\ref{line:dv-add-value} adds it to $I$, and
line~\ref{line:dv-deactivate} marks exactly the agents holding this
value as no longer remaining. Hence each iteration removes one distinct
input value, starting from the largest. After exactly $q$ iterations,
no agent remains, so the next test in line~\ref{line:dv-remaining-loop}
terminates the algorithm. Therefore the set returned in
line~\ref{line:dv-return} is precisely the set of input values initially
present.

For the complexity, the first phase performs $O(1+B_{\max})$ calls to
\Flood. Each extraction performs one call in
line~\ref{line:dv-remaining-loop} and one call per bit in
line~\ref{line:dv-test-one}, hence $O(1+B_{\max})$ calls to \Flood per
extracted value. Since there are $q$ extracted values and each call to
\Flood lasts $U-1$ rounds, the total number of rounds is
$ O(Uq(1+B_{\max})) \leq O(Un(1+B_{\max}))$.
\end{proof}

Notice that, when \(U=O(n)\), Algorithm~\ref{alg:distinct-values} nearly
matches the lower bound of Section~\ref{sec:lower-bound-bc}. Indeed, if
the inputs are drawn from a universe of size \(N\), then
\(B_{\max}=O(\log N)\), and the algorithm uses
\(O(nq\log N)\subseteq O(n^2\log N)\) rounds. The lower bound is
\(\Omega(n^2\log(N/n)/\log n)\) rounds. Thus,
when \(N\) is polynomially related to \(n\), the upper and lower bounds
differ only by logarithmic factors. In particular, the quadratic
dependence on \(n\) is essentially tight.

\section{Computing with Multiple Leaders} \label{sec:leader-anonymous-bc}

We first compute the \ProblemName{Input Frequency} function without a leader and then
derive the \ProblemName{Input Multiset} function when the number \(k\) of leader-flagged
agents is known.

\subsection{Computing \ProblemName{Input Frequency}}
\label{sec:input-frequency}

We give a terminating algorithm for the \ProblemName{Input Frequency}
problem. All agents know a common upper bound $U\geq n$, but no
leader is required.

\paragraph{Overview.}
The agents first compute the set of input values occurring in the
network and assign one common label to each value. Agents with the
same label form a \emph{class}. If the current classes are
$P_1,\ldots,P_m$, let
\[
    c=(c_1,\ldots,c_m),
    \qquad c_i=|P_i|,
\]
be their size vector. The algorithm learns homogeneous linear equations
satisfied by $c$. Once it has $m-1$ independent equations, their common
solution space is the line spanned by $c$. Normalizing any nonzero
solution so that its coordinates sum to $1$ yields
\[
    \frac{c}{n}
    =\left(\frac{c_1}{n},\ldots,\frac{c_m}{n}\right),
\]
the frequencies of the current classes.

To obtain one equation, the algorithm chooses a nonempty proper set $S$
of class labels. In one test round, agents whose labels belong to $S$
send $1$, and all other agents send $0$. Each agent counts its
neighbors that sent the opposite bit. Counting the edges between the
two sides from either side gives a homogeneous equation.

Agents in the same class may obtain different counts. The algorithm
therefore splits each class according to the count observed by its
members. Previously learned equations are transferred to the refined
classes and remain valid. The set $S$ is chosen so that the new equation
is independent of the previous ones. We first describe this operation,
then show that a suitable $S$ always exists, and finally give the complete
algorithm.

\subsubsection{Classes and constraints}

At any step of the algorithm, the class labels are the consecutive
integers $1,\ldots,m$, and every label is held by at least one agent.
We call such labels \emph{compact}. The corresponding classes are $P_i=\{v\in V:\mathit{label}(v)=i\},$ with $ i\in\{1,\ldots,m\}$. Initially, each class consists of all agents having the same input value. Since subsequent refinements may split classes but never merge them, every current class is contained in
a unique initial class and therefore has a well-defined original input value.

A vector $a\in\mathbb R^m$ is a \emph{constraint} for the current
partition if $a\cdot c=0$.

The algorithm maintains a linearly independent list
$\mathcal C=(a^1,\ldots,a^r)$
of such constraints. Initially, $\mathcal C$ is empty. If $r=m-1$,
then the vectors in $\mathcal C$ span $c^\perp$. Thus every nonzero
solution of the equations $a^j\cdot y=0$ is a scalar multiple of $c$,
and normalizing it by the sum of its coordinates gives $c/n$.
\subsubsection{Obtaining one constraint with \ConstraintRound}

The procedure that creates a new constraint is given in Algorithm~\ref{alg:bc-constraint-round}. Consider one invocation of \ConstraintRound{} with a nonempty proper set \(S\subsetneq\{1,\ldots,m\}\). We call the classes present at the beginning of the invocation the \emph{old classes}, and the classes produced by the invocation the \emph{refined classes}.

The invocation begins with a single \emph{test round}. In this round, an agent sends \(1\) if its old label belongs to \(S\), and sends \(0\) otherwise. Here \(\mathit{oldLabel}(u)\) denotes the label held by \(u\) at the beginning of the invocation. For an agent \(u\), let
\[
\begin{aligned}
    d_1(u)&=\#\{\text{neighbors of \(u\) that sent \(1\)}\},\\
    d_0(u)&=\#\{\text{neighbors of \(u\) that sent \(0\)}\}.
\end{aligned}
\]
The \emph{opposite-side count} of \(u\) is
\[
    h(u)=
    \begin{cases}
        d_0(u),&\mathit{oldLabel}(u)\in S,\\
        d_1(u),&\mathit{oldLabel}(u)\notin S.
    \end{cases}
\]
Thus \(h(u)\) is the number of neighbors of \(u\) whose old labels lie on the opposite side of the partition induced by \(S\). The test round and the computation of the opposite-side counts are implemented in lines~\ref{line:cr-send}--\ref{line:cr-relevant-count}.

After the test round, \ConstraintRound{} splits each old class according to the opposite-side counts observed by its members. For each old class \(x\), all agents invoke \DistinctValues{}. An agent in class \(x\) contributes the encoded value \(h+2\), while every agent outside class \(x\) contributes the sentinel value \(1\). Since \(h\geq0\), the sentinel is distinct from every encoded count. After removing \(1\) from the set returned by \DistinctValues{}, all agents keep the remaining encoded values in sorted order and subtract \(2\) only when recording the corresponding opposite-side counts.

The smallest opposite-side count occurring in old class \(x\) retains label \(x\), while every other count receives a fresh label. Fresh labels are assigned in a fixed order. Hence all agents construct the same refined partition and the same parent map. For every refined label \(j\), \(\mathit{parent}[j]\) is the label of the old class from which it originates, and \(h_j\) is the opposite-side count shared by all agents in the corresponding refined class. The encoded counts are collected and the old classes are refined in lines~\ref{line:cr-encode-count}--\ref{line:cr-relabel}.

Let \(c'\) be the refined class-size vector, and write
$\pi(j)=\mathit{parent}[j]$.
Every previously known constraint \(\alpha\in\mathcal C\) is lifted to the refined classes by defining $\widehat{\alpha}_j=\alpha_{\pi(j)}$.

Since the refined classes with parent \(i\) partition old class \(i\), every lifted constraint remains valid for \(c'\).

The procedure defines a new vector \(a\in\mathbb R^{m'}\) by
\[
    a_j=
    \begin{cases}
        h_j,&\pi(j)\in S,\\
        -h_j,&\pi(j)\notin S.
    \end{cases}
\]
The old constraints are lifted in line~\ref{line:cr-lift}, and the new constraint is defined and added in lines~\ref{line:cr-define-constraint}--\ref{line:cr-append-constraint}.

Note that the sum
    $\sum_{j:\pi(j)\in S} h_jc'_j$
counts the edges crossing the partition from the \(S\) side, whereas $\sum_{j:\pi(j)\notin S} h_jc'_j$ counts the same edges from the other side. The two sums are therefore equal, and hence $a\cdot c'=0$.
Thus \(a\) is a valid constraint for the refined class-size vector.

Validity alone is not sufficient: the new constraint must also be independent of the lifted constraints already known. Otherwise, the invocation would add no new information about the class sizes, and the rank of the constraint system would not increase.

To express the condition that guarantees independence, define
\[
    K_S=
    \left\{
        z\in\mathbb R^m:
        z_i\geq0\text{ for }i\in S,
        \quad
        z_i\leq0\text{ for }i\notin S
    \right\}.
\]
We call a nonempty proper set \(S \subsetneq \{1,\ldots,m\}\) a \emph{progress set} for \(\mathcal C\) if
 $\operatorname{span}(\mathcal C)\cap K_S=\{0\}$.

The name reflects the role of this condition: choosing a progress set ensures that the constraint produced by \ConstraintRound{} is independent of the lifted old constraints and therefore increases the rank of the constraint system by one.

The next lemma establishes the properties of one invocation of \ConstraintRound{}. In particular, it shows that the procedure correctly refines the classes, preserves the previously known constraints under refinement, and adds one new independent constraint whenever \(S\) is a progress set. In Section~\ref{sec:wemakeprogress}, we prove that a progress set always exists as long as fewer than \(m-1\) independent constraints are known, and we explain how all agents can select the same progress set using only local computation.

\LinesNumbered

\begin{algorithm}[H]
\caption{\ConstraintRound}
\label{alg:bc-constraint-round}

\KwInput{A nonempty proper set \(S\subsetneq\{1,\ldots,m\}\), a common upper bound \(U\geq n\), the
current number \(m\) of classes, the constraint list \(\mathcal C\), and
the local label \(\mathit{label}\)}
\KwOutput{The refined label, the new number of classes, the augmented
constraint list, and the parent map}

\Fn{\AlgConstraintRound{$S,U,m,\mathcal C,\mathit{label}$}}{
  $\mathit{oldLabel}\leftarrow\mathit{label}$;
  $m'\leftarrow m$\label{line:cr-initialize}\;

\tcp{Test round: agents whose old labels belong to \(S\) send \(1\), while all other agents send \(0\)}
  send \(1\) iff \(\mathit{oldLabel}\in S\), and send \(0\)
  otherwise\label{line:cr-send}\;

  $d_1\leftarrow$ number of neighbors that sent \(1\);
  $d_0\leftarrow$ number of neighbors that sent \(0\)
  \label{line:cr-observe}\;

  $h\leftarrow d_0$ if \(\mathit{oldLabel}\in S\), and
  $h\leftarrow d_1$ otherwise\label{line:cr-relevant-count}\;

  $o\leftarrow h+2$\label{line:cr-encode-count}\;

  \For{$x\leftarrow1$ \KwTo $m$}{
    $z\leftarrow o$ if \(\mathit{oldLabel}=x\), and
    $z\leftarrow1$ otherwise\label{line:cr-select-input}\;

    $O_x\leftarrow\AlgDistinctValues(z,U)\setminus\{1\}$
    \label{line:cr-discover-counts}\;

    write \(O_x=\{o_{x,1}<\cdots<o_{x,q_x}\}\)
    \label{line:cr-order-counts}\;

    $\mathit{parent}[x]\leftarrow x$;
    $h_x\leftarrow o_{x,1}-2$
    \label{line:cr-retain-label}\;

    \For{$j\leftarrow2$ \KwTo $q_x$}{
      $m'\leftarrow m'+1$
      \label{line:cr-create-label}\;

      $\mathit{parent}[m']\leftarrow x$;
      $h_{m'}\leftarrow o_{x,j}-2$
      \label{line:cr-record-child}\;

      \If{$\mathit{oldLabel}=x$ and $o=o_{x,j}$}{
        $\mathit{label}\leftarrow m'$
        \label{line:cr-relabel}\;
      }
    }
  }
  \(\mathcal C' \leftarrow []\)\;
  for each \(\alpha\in\mathcal C\) append its lift
  \(\widehat\alpha\in\mathbb R^{m'}\) to \(\mathcal C'\), where
  \(\widehat\alpha_j=\alpha_{\mathit{parent}[j]}\)
  \label{line:cr-lift}\;

  define \(a\in\mathbb R^{m'}\) by
  \(a_j=h_j\) if \(\mathit{parent}[j]\in S\), and
  \(a_j=-h_j\) otherwise
  \label{line:cr-define-constraint}\;

  append \(a\) to \(\mathcal C'\)
  \label{line:cr-append-constraint}\;

  \Return{$(\mathit{label},m',\mathcal C',\mathit{parent})$}
  \label{line:cr-return}\;
}
\end{algorithm}

\begin{lemma}[Correctness of one constraint round]
\label{lem:bc-constraint-round-correct}
Assume that the current labels are compact. Let \(c\in\mathbb R_{>0}^m\)
be the current class-size vector, and let the vectors in \(\mathcal C\)
be linearly independent constraints for \(c\). If \(S\) is a progress
set for \(\mathcal C\), then after one invocation of \ConstraintRound:
\begin{enumerate}
    \item the new labels are compact, and two agents have the same
          new label exactly when they had the same old label and the
          same opposite-side count;
    \item for every refined class \(j\), \(\mathit{parent}[j]\) is the old
          label of that class and \(h_j\) is the common opposite-side count of
          its agents;
    \item the returned list \(\mathcal C'\) contains
          \(|\mathcal C|+1\) linearly independent constraints for the
          refined class-size vector \(c'\).
\end{enumerate}
\end{lemma}

\begin{proof}
Fix an old class \(x\). By lines~\ref{line:cr-relevant-count} and
\ref{line:cr-encode-count}, every agent computes its opposite-side
count \(h\) and encodes it as \(o=h+2\). In
line~\ref{line:cr-select-input}, agents in class \(x\) contribute
their encoded values, while all other agents contribute the sentinel
value \(1\). Since every encoded value is at least \(2\), the sentinel
cannot coincide with an encoded count. Consequently,
line~\ref{line:cr-discover-counts} returns exactly the encoded counts
occurring in class \(x\), after the sentinel has been removed. Thus all
agents obtain the same nonempty set \(O_x\), ordered in
line~\ref{line:cr-order-counts}.

Line~\ref{line:cr-retain-label} records the original label \(x\) as the
retained label for the smallest count occurring in class \(x\). Agents with
that smallest count keep label \(x\), because they are not relabeled later in
the loop. For every remaining count,
lines~\ref{line:cr-create-label} and \ref{line:cr-record-child} create a
fresh consecutive label and record its old parent and opposite-side
count. Line~\ref{line:cr-relabel} assigns this fresh label exactly to the
agents of class \(x\) having the corresponding count.

It follows that each distinct count occurring in old class \(x\)
corresponds to exactly one refined class. The original label \(x\)
remains in use, and every additional label is introduced consecutively.
Since this construction is performed for every old class, all labels
\(1,\ldots,m'\) occur. Moreover, two agents receive the same new
label exactly when their old labels and opposite-side counts agree.
This proves part~(1).
The construction also records the old parent and the corresponding decoded
opposite-side count for each retained or fresh label, proving part~(2).

Let $\pi(j)=\mathit{parent}[j]$
be the parent of refined class \(j\). In
line~\ref{line:cr-lift}, every old constraint
\(\alpha\in\mathcal C\) is replaced by the vector
\(\widehat\alpha\in\mathbb R^{m'}\) defined by $\widehat\alpha_j=\alpha_{\pi(j)}$.

Since the children of old class \(i\) form a partition of that class,
\[
\begin{aligned}
    \widehat\alpha\cdot c'
    &= \sum_{j=1}^{m'}\alpha_{\pi(j)}c'_j
    &=
    \sum_{i=1}^{m}\alpha_i
        \sum_{j:\pi(j)=i}c'_j
    &=
    \sum_{i=1}^{m}\alpha_i c_i
    &=
    0.
\end{aligned}
\]
Thus every lifted vector remains a constraint for \(c'\).

The lifting map is injective. Indeed, by
line~\ref{line:cr-retain-label}, every old class has at least one
refined child. Therefore, if two old vectors have the same lift, their
coordinates agree on at least one child of every old class, and hence
the old vectors are equal. It follows that the lifted constraints remain
linearly independent.

We next consider the vector \(a\) defined in
line~\ref{line:cr-define-constraint}. By
lines~\ref{line:cr-send}--\ref{line:cr-relevant-count}, an agent whose
old label lies in \(S\) sets \(h\) equal to the number of neighbors
outside \(S\), while an agent whose old label lies outside \(S\) sets
\(h\) equal to the number of neighbors inside \(S\). Therefore,
summing the relevant counts over the two sides counts the same crossing
edges:
\[
    \sum_{j:\pi(j)\in S} h_jc'_j
    =
    \sum_{j:\pi(j)\notin S} h_jc'_j.
\]
The signs assigned in line~\ref{line:cr-define-constraint} consequently
give $a\cdot c'=0$. Hence \(a\) is also a constraint for \(c'\).

Because \(S\) is nonempty and proper, both sides of the induced
partition contain at least one agent. The communication graph in the
test round is connected, so at least one edge crosses the partition.
By lines~\ref{line:cr-observe} and \ref{line:cr-relevant-count}, at least
one relevant count is therefore positive. Consequently, the vector
\(a\) defined in line~\ref{line:cr-define-constraint} is nonzero.

It remains to prove that \(a\) is independent of the lifted old constraints.
Suppose, for contradiction, that \(a\) belongs to their span. Then there
exist coefficients \(\lambda_1,\ldots,\lambda_r\) such that
\(a=\sum_{\ell=1}^r\lambda_\ell\widehat{\alpha^\ell}\), where
\(\mathcal C=(\alpha^1,\ldots,\alpha^r)\). Define \(b=\sum_{\ell=1}^r\lambda_\ell \alpha^\ell\).
By construction, \(b\in\operatorname{span}(\mathcal C)\). Moreover, for every
refined class \(j\), we have \(a_j=b_{\pi(j)}\), because \(\widehat{\alpha^\ell}_j=\alpha^\ell_{\pi(j)}\).
Since every old class has at least one refined child and \(a\neq0\), it follows that \(b\neq0\).

If \(i\in S\), then every refined child \(j\) of \(i\) satisfies
\(b_i=a_j=h_j\geq0\). If \(i\notin S\), then every refined child \(j\) of \(i\)
satisfies \(b_i=a_j=-h_j\leq0\). Hence \(0\neq b\in\operatorname{span}(\mathcal C)\cap K_S\),
contradicting the assumption that \(S\) is a progress set. Therefore,
\(a\) is independent of the lifted old constraints.

\end{proof}

\begin{lemma}
\label{lem:bc-constraint-round-complexity}
Suppose that an invocation of \ConstraintRound{} starts with \(m\)
classes and creates \(s\) fresh labels. It uses $O\bigl(U(m+s)\log n\bigr)$ rounds.
\end{lemma}

\begin{proof}
Lines~\ref{line:cr-send} and \ref{line:cr-observe} execute one
communication round. Computing the relevant count and its encoding in
lines~\ref{line:cr-relevant-count} and \ref{line:cr-encode-count} is
local.

For an old class \(x\), let \(q_x\) be the number of distinct
opposite-side counts occurring in that class. By
line~\ref{line:cr-select-input}, the invocation of \DistinctValues{} in
line~\ref{line:cr-discover-counts} receives at most \(q_x+1\) distinct
values: the \(q_x\) encoded counts and the sentinel value \(1\). Each opposite-side count is at most \(n-1\). Hence, by
line~\ref{line:cr-encode-count}, every encoded value is at most \(n+1\)
and has bit length \(O(\log n)\). Lemma~\ref{lem:bc-distinct-values}
therefore bounds the invocation in
line~\ref{line:cr-discover-counts} by
$O\bigl(U(q_x+1)\log n\bigr)$
rounds.

Lines~\ref{line:cr-retain-label}--\ref{line:cr-relabel} create one
refined class for each of the \(q_x\) counts occurring in old class
\(x\). Consequently, the total number of refined classes is
$\sum_{x=1}^{m}q_x$.

Since the procedure starts with \(m\) labels and creates \(s\) fresh
ones, the final number of classes is \(m+s\). Therefore, $\sum_{x=1}^{m}q_x=m+s$.

Summing the costs of line~\ref{line:cr-discover-counts} over all old
classes gives:
\[
\begin{aligned}
    \sum_{x=1}^{m}O\bigl(U(q_x+1)\log n\bigr)
    &=
    O\left(
        U\log n
        \sum_{x=1}^{m}(q_x+1)
    \right)\\
    &=
    O\bigl(U(2m+s)\log n\bigr)\\
    &=
    O\bigl(U(m+s)\log n\bigr).
\end{aligned}
\]
The lifting, construction of the new constraint, and all label
operations in lines~\ref{line:cr-order-counts}--\ref{line:cr-return}
are local. The single test round is absorbed by the stated bound.
\end{proof}

\begin{lemma}
\label{lem:bc-total-refinements}
If the initial partition has $q$ classes, all calls to \ConstraintRound{}
create at most $n-q$ fresh labels in total.
\end{lemma}

\begin{proof}
Each fresh label denotes a new nonempty class, classes are never merged,
and there can be at most $n$ nonempty classes.
\end{proof}

\subsubsection{Choosing a set that guarantees independence}\label{sec:wemakeprogress}

The agents do not know the class-size vector $c$. They therefore
choose $S$ using only the common constraint list $\mathcal C$. We first
show that a progress set exists whenever more constraints are needed,
and then show how to find one locally. Recall that $K_S=
    \left\{
        z\in\mathbb R^m:
        z_i\geq0\text{ for }i\in S,
        z_i\leq0\text{ for }i\notin S
    \right\}$.

\begin{lemma}[Existence of a progress set]
\label{lem:bc-progress-set-exists}
\label{lem:bc-sign-constraints-span-hyperplane}
Let $c\in\mathbb R_{>0}^m$, and let $L\subseteq c^\perp$. If
$\dim L<m-1$, then there exists a nonempty proper set
$S\subsetneq\{1,\ldots,m\}$ such that $L\cap K_S=\{0\}$.
\end{lemma}

\begin{proof}
Since $\dim L<m-1$, the space $L^\perp$ has dimension at least two. It
contains $c$, so choose $y\in L^\perp$ that is not a scalar multiple of
$c$.

The ratios $y_i/c_i$ are therefore not all equal. Choose $t$ strictly
between their minimum and maximum and different from every ratio, and set $S=\{i:y_i/c_i>t\}$. Then $S$ is nonempty and proper.

Suppose that a nonzero vector $a\in L\cap K_S$ exists. Since
$a\cdot c=0$ and every coordinate of $c$ is positive, $a$ has both a
positive and a negative coordinate. By the definitions of $S$ and
$K_S$, every term in
\[
    \sum_{i=1}^{m}a_i c_i\left(\frac{y_i}{c_i}-t\right)
\]
is nonnegative, and at least one is positive. The sum is therefore
positive. But it is also equal to
$a\cdot y-t(a\cdot c)=0$, because $a\in L$ and $y,c\in L^\perp$. This contradiction proves the
claim.
\end{proof}

\begin{lemma}[Testing a candidate set]
\label{lem:bc-lp-sign-vector}
Let \(L\subseteq\mathbb R^m\). For every nonempty proper set
\(S\subsetneq\{1,\ldots,m\}\), one can decide by linear programming whether
$L\cap K_S$ contains a nonzero vector.
\end{lemma}

\begin{proof} Let $L$ be given by a basis
$b^1,\ldots,b^r$. Write $z=\sum_{j=1}^{r}\theta_j b^j$ and use $\theta_1,\ldots,\theta_r$ as the variables. Impose
    $z_i\geq0\quad(i\in S)$,
     and
    $z_i\leq0\quad(i\notin S)$, together with
    $\sum_{i\in S}z_i-
    \sum_{i\notin S}z_i\geq1$.

Under the sign constraints, the last left-hand side is
$\sum_i|z_i|$. Hence feasibility implies that $z$ is a nonzero vector in
$L\cap K_S$. Conversely, any nonzero vector in $L\cap K_S$ can be
scaled by a positive constant to satisfy the last inequality.
\end{proof}

\begin{corollary}[Selecting the next progress set]
\label{cor:bc-missing-sign-pattern}
Let $c\in\mathbb R_{>0}^m$, and let the vectors in $\mathcal C$ be
linearly independent constraints for $c$. If $|\mathcal C|<m-1$, then
all agents can select the same progress set without communication.
\end{corollary}

\begin{proof}
Apply Lemma~\ref{lem:bc-progress-set-exists} to
$L=\operatorname{span}(\mathcal C)$. A progress set exists because
$\dim L=|\mathcal C|<m-1$. All agents know the same $m$ and
$\mathcal C$, enumerate the nonempty proper subsets of
$\{1,\ldots,m\}$ in the same fixed order, and use
Lemma~\ref{lem:bc-lp-sign-vector} to select the first progress set. These
are local computations and use no rounds.
\end{proof}

\subsubsection{The complete \ProblemName{Input Frequency} algorithm}

The \InputFrequency{} procedure is given in Algorithm~\ref{alg:bc-input-frequency}.
The agents use \DistinctValues to obtain the set of input values and assign initial labels according to their increasing order. The array $\mathit{origin}$ stores the original input associated with each current class; whenever a class is split, each child inherits its parent's value.

The procedure then repeatedly calls \ConstraintRound to obtain new constraints until the accumulated constraints are sufficient to compute the frequencies of the input.

\begin{algorithm}[H]
\caption{\InputFrequency}
\label{alg:bc-input-frequency}

\KwInput{A positive integer input $x$ and a common upper bound $U\geq n$}
\KwOutput{The frequency of every input value occurring in the network}

\Fn{\AlgInputFrequency{$x,U$}}{
  $I\leftarrow\AlgDistinctValues(x,U)$\;
  write $I=\{w_1<\cdots<w_m\}$\;
  $\mathit{label}\leftarrow i$ such that $x=w_i$\;

  \For{$i\leftarrow1$ \KwTo $m$}{
    $\mathit{origin}[i]\leftarrow w_i$\;
  }
  $\mathcal C\leftarrow\emptyset$\;

  \While{$|\mathcal C|<m-1$}{
    choose the first nonempty proper set
    $S\subsetneq\{1,\ldots,m\}$ in a fixed common order such that
    $\operatorname{span}(\mathcal C)\cap K_S=\{0\}$\;

    $\mathit{oldOrigin}\leftarrow\mathit{origin}$\;
    $(\mathit{label},m,\mathcal C,\mathit{parent})\leftarrow
      \AlgConstraintRound(S,U,m,\mathcal C,\mathit{label})$\;

    \For{$j\leftarrow1$ \KwTo $m$}{
      $\mathit{origin}[j]\leftarrow
       \mathit{oldOrigin}[\mathit{parent}[j]]$\;
    }
  }

  solve for $y=(y_1,\ldots,y_m)$ the system
  $\alpha\cdot y=0$ for every $\alpha\in\mathcal C$ and
  $\sum_{i=1}^{m}y_i=1$\;

  \ForEach{$w\in I$}{
    $f_w\leftarrow
      \displaystyle\sum_{j:\mathit{origin}[j]=w}y_j$\;
  }
  \Return{the map $w\mapsto f_w$ for $w\in I$}\;
}
\end{algorithm}

\begin{theorem}
\label{thm:anonymous-bc-frequency}
Assume that all agents know a common upper bound $U\geq n$.
The algorithm \InputFrequency{} computes the \ProblemName{Input Frequency} function with
termination in anonymous \(1\)-interval-connected dynamic networks. It uses
$O\bigl(UnB_{\max}+Un^2\log n\bigr)$
rounds.
\end{theorem}

\begin{proof}

By Lemma~\ref{lem:bc-distinct-values}, all agents obtain the same set
$I$. Its increasing order gives compact initial labels, and
$\mathit{origin}$ records the correct input for every initial class.
Initially, $\mathcal C$ is empty.

At the beginning of every loop iteration, the labels are compact, every
class has the correct value in $\mathit{origin}$, and the vectors in
$\mathcal C$ are linearly independent constraints for the current
class-size vector. If the loop guard holds,
Corollary~\ref{cor:bc-missing-sign-pattern} guarantees that the required
set $S$ exists and that all agents select the same one. By
Lemma~\ref{lem:bc-constraint-round-correct}, the call to
\ConstraintRound{} preserves compactness and the old constraints, adds one
independent constraint, and returns the correct parent map. The update
of $\mathit{origin}$ is therefore also correct. Thus these properties
hold throughout the execution.

Each iteration increases $|\mathcal C|$ by one. Since the vectors in
$\mathcal C$ are independent and orthogonal to the nonzero current
class-size vector, $|\mathcal C|\leq m-1\leq n-1$.

Every invocation of \ConstraintRound{} terminates, by Lemma \ref{lem:bc-constraint-round-complexity}, so the loop performs at most $n-1$ iterations. At termination, its guard is false; hence
$|\mathcal C|=m-1$.

Let $c$ be the final class-size vector. The vectors in $\mathcal C$ form
a basis of $c^\perp$, so the homogeneous equations in the final system
have solution space $\operatorname{span}\{c\}$. Since
$\sum_i c_i=n$, the additional equation $\sum_i y_i=1$ gives the unique
solution $y=c/n$. For every original input value $w$, the current
classes with $\mathit{origin}[j]=w$ partition exactly the agents whose
input is $w$. Therefore
\[
    f_w
    =\sum_{j:\mathit{origin}[j]=w}\frac{c_j}{n}
    =\frac{|\{v\in V:\lambda(v)=w\}|}{n}.
\]
Thus the output is correct.

For the complexity, let $q=|I|$. The initial call to \DistinctValues{}
uses $O\bigl(Uq(1+B_{\max})\bigr)=O(UnB_{\max})$ rounds, since inputs are positive integers and $B_{\max}\geq1$.
For iteration $i$, let $m_i$ be the number of classes at its start and
$s_i$ the number of fresh labels it creates. There are at most $n-1$
iterations, $m_i\leq n$, and $\sum_i s_i\leq n-q$
by Lemma~\ref{lem:bc-total-refinements}. Hence
Lemma~\ref{lem:bc-constraint-round-complexity} gives
$\sum_i O\bigl(U(m_i+s_i)\log n\bigr)
    =O\bigl(Un^2\log n\bigr)$. Adding the initial call proves the stated bound.
\end{proof}

\subsection{\ProblemName{Input Multiset} with a Known Number of Leaders}
\label{sec:frequency-to-multiset}

Relative frequencies become exact multiplicities once the multiplicity of some
distinguished population is known.  In this section we assume that exactly
\(k\geq1\) agents have their leader flag set, and that \(k\) is known to all
agents.

The \InputMultiset{} procedure is given in Algorithm~\ref{alg:bc-input-multiset}.

An agent with input $x$ and leader flag $b\in\{0,1\}$ encodes the pair
as $z=2x+b$.  This encoding is injective, and $z$ is odd exactly for leaders.
Note that leaders need not have the same input.
After running \InputFrequency{} on the encoded values, let $f_z$ be the
frequency of $z$, with $f_z=0$ for absent values. The total frequency of
odd encoded values is
$f_{\mathrm L}=\sum_{z\text{ odd}}f_z=\frac{k}{n}$.

Thus all agents recover $n=k/f_{\mathrm L}$. For every original input
$a$, its multiplicity is
$n\bigl(f_{2a}+f_{2a+1}\bigr)$.

\begin{algorithm}[H]
\caption{\procname{InputMultisetWithMultiLeaders}}
\label{alg:bc-input-multiset}

\KwInput{A positive integer input $x$, the local leader flag
\(\mathit{isLeader}\), a common upper bound $U\geq n$,
and the known number $k\geq1$ of leaders}
\KwOutput{The multiplicity of every input value occurring in the network}

\Fn{\AlgInputMultiset{$x,\mathit{isLeader},U,k$}}{
  \eIf{$\mathit{isLeader}$}{
    $z\leftarrow2x+1$\;
  }{
    $z\leftarrow2x$\;
  }

  $f\leftarrow\AlgInputFrequency(z,U)$\;
  $Z\leftarrow\operatorname{dom}(f)$\;
  treat $f_z$ as $0$ for every $z\notin Z$\;

  $f_{\mathrm L}\leftarrow
    \displaystyle\sum_{z\in Z:\,z\text{ odd}}f_z$\;
  $\widehat n\leftarrow k/f_{\mathrm L}$\;
  $I\leftarrow\{\lfloor z/2\rfloor:z\in Z\}$\;

  \ForEach{$a\in I$}{
    $C_a\leftarrow\widehat n\bigl(f_{2a}+f_{2a+1}\bigr)$\;
  }
  \Return{the map $a\mapsto C_a$}\;
}
\end{algorithm}

\begin{theorem}
\label{thm:anonymous-bc-input-multiset}
In anonymous \(1\)-interval-connected dynamic networks with exactly
$k\geq1$ leaders, where $k$ is known to all agents, \procname{InputMultisetWithMultiLeaders}
computes the \ProblemName{Input Multiset} function with termination, assuming a common
upper bound $U\geq n$. It uses
$O\bigl(UnB_{\max}+Un^2\log n\bigr)$ rounds.
\end{theorem}

\begin{proof}
Exactly the leaders have odd encoded values, so $f_{\mathrm L}=\frac{k}{n}>0$ and hence $\widehat n = \frac{k}{f_{\mathrm L}} = n$.
For every original input $a$, the agents holding $a$ are exactly
those with encoded value $2a$ or $2a+1$. Therefore $\widehat n\bigl(f_{2a}+f_{2a+1}\bigr)
    =|\{v\in V:\lambda(v)=a\}|$.

The output is correct, and termination follows from
Theorem~\ref{thm:anonymous-bc-frequency}. Since the encoding $2x+b$
increases the input bit length by at most one, the same theorem gives the
stated complexity.
\end{proof}

The unique-leader setting is the special case $k=1$.

\begin{corollary}
\label{cor:anonymous-bc-classical-counting}
In anonymous \(1\)-interval-connected dynamic networks with exactly \(k\geq1\)
leaders, where \(k\) is known to all agents, the \ProblemName{Counting} function can be computed
with termination in
$O\bigl(Un^2\log n\bigr)$ rounds, assuming a common upper bound $U\geq n$.
\end{corollary}

\begin{proof}
Each agent ignores its original input and invokes
\procname{InputMultisetWithMultiLeaders} with synthetic input \(1\), common bound
\(U\), known leader count \(k\), and its original leader flag. The multiplicity
of the synthetic value \(1\) is \(n\). Since the
input bit length is constant, the term $O(UnB_{\max})$ is absorbed by
$O(Un^2\log n)$.
\end{proof}

\subsection{Computing an Upper Bound with a Known Number of Leaders and a Local Degree Oracle}\label{sec:upperbound}

We now present a procedure that computes a common upper bound on the network
size in a \(1\)-interval-connected dynamic network with exactly $k \geq 1$ leaders, $k$ known to all agents, and a
local degree oracle. The procedure, called \UpperBound, is given in Algorithm~\ref{alg:upper-bound}.

The algorithm considers successive trials indexed by an integer $q$.
In the trial with parameter $q$, the candidate upper bound is
$U=k(q+k)^q$. Initially, exactly the $k$ leaders are reached. Their signal
then propagates for $q$ rounds, but a reached agent forwards the signal
only when its degree in the current round is at most $q+k-1$.

After the propagation phase, every agent invokes
\(\Flood(U,\lnot\mathit{reached})\). Thus, an unreached agent supplies $1$,
while a reached agent supplies $0$. If some agent is unreached, \Flood
returns $1$ at every agent and the candidate is rejected. Otherwise, \Flood
returns $0$ at every agent and all agents return $U$.

\LinesNumbered
\begin{algorithm}[H]
\caption{\UpperBound}
\label{alg:upper-bound}

\KwInput{The local leader flag \(\mathit{isLeader}\), the known number \(k\geq1\) of leaders and access to the local degree oracle}
\KwOutput{A common integer $U$ satisfying $U\geq n$}

\Fn{\AlgUpperBound{$\mathit{isLeader}$, $k$}}{
  \For{$q \leftarrow 0,1,2,\ldots$}{
    $U \leftarrow k(q+k)^q$\label{line:ub-candidate}\;

    $\mathit{reached} \leftarrow \mathit{isLeader}$\label{line:ub-init}\;

    \For{$t \leftarrow 1$ \KwTo $q$}{
      $d \leftarrow$ current degree returned by the local degree oracle\label{line:ub-degree}\;

      send $1$ iff $\mathit{reached}$ and $d\leq q+k-1$; otherwise send $0$\label{line:ub-send}\;

      \If{at least one neighbor sent $1$}{
        $\mathit{reached} \leftarrow \boolname{true}$\label{line:ub-update}\;
      }
    }

    $\mathit{failed} \leftarrow $\AlgFlood{$U,\lnot\mathit{reached}$}\label{line:ub-flood}\;

    \If{$\lnot\mathit{failed}$}{
      \Return{$U$}\label{line:ub-return}\;
    }
  }
}
\end{algorithm}

The local degree oracle is queried at line~\ref{line:ub-degree}, and the returned
degree is tested at line~\ref{line:ub-send}. An agent whose degree is
larger than \(q+k-1\) does not immediately reject the candidate; it simply
refrains from forwarding the signal. This restriction is sufficient to bound
the number of reached agents by \(U=k(q+k)^q\).

The invocation of \Flood at line~\ref{line:ub-flood} has duration $U$, rather than $U-1$. At this point, $U$ is not yet known to be an upper bound on the entire network, so Corollary~\ref{cor:flood-upper-bound} cannot be applied. Instead, the proof shows that at most $U$ agents begin this invocation of \Flood with value $0$, allowing us to apply Lemma~\ref{lem:flood-correctness} directly.

\begin{lemma}\label{lem:bc-upper-bound-correct}
In anonymous \(1\)-interval-connected dynamic networks with exactly $k \geq 1$ leaders,
with $k$ known to all agents and with a local degree oracle, all
agents executing \UpperBound terminate in the same round with the
same value $U$, and this value satisfies $n\leq U$.
\end{lemma}

\begin{proof}
Fix a trial with parameter $q$. For every $i\in\{0,\ldots,q\}$,
let $R_i$ be the set of agents whose variable $\mathit{reached}$ is
true after the first $i$ rounds of the propagation phase. By line~\ref{line:ub-init},
$R_0$ contains exactly the leaders, and hence $|R_0|=k$.

We first show that $|R_i|\leq k(q+k)^i$ for every $i\in\{0,\ldots,q\}$.
Consider propagation round $i$, where $1\leq i\leq q$. By synchrony,
only agents in $R_{i-1}$ can send $1$ in this round. By line~\ref{line:ub-send},
every agent that sends $1$ has degree at most $q+k-1$. Therefore, each sender
can cause at most $q+k-1$ agents to become newly reached, and at most $(q+k-1)|R_{i-1}|$
agents become newly reached in total. It follows that $|R_i|\leq(q+k)|R_{i-1}|$.
Since $|R_0|=k$, induction gives $|R_i|\leq k(q+k)^i$.

In particular, at the end of the propagation phase we have $|R_q|\leq k(q+k)^q$.
By line~\ref{line:ub-candidate}, the candidate is $U=k(q+k)^q$, and therefore $|R_q|\leq U$.

Consider now the invocation of \Flood at line~\ref{line:ub-flood}. Let $S_0$ be
the set of agents whose initial value in this invocation is $1$. Since the
initial value is $\lnot\mathit{reached}$, we have $S_0=V\setminus R_q$.

If $S_0$ is empty, all agents invoke \Flood with value $0$. By part~(1) of
Lemma~\ref{lem:flood-correctness}, every agent returns $0$.

Suppose instead that $S_0$ is nonempty. The agents that invoke \Flood with
value $0$ are exactly the agents in $R_q$. Consequently, $|V\setminus S_0|=|R_q|\leq U$.
Since the duration of \Flood is $U$, part~(2) of Lemma~\ref{lem:flood-correctness}
implies that every agent returns $1$.

Thus, the call to \Flood returns the same value at every agent and, because
the duration of the trial is fixed by \(q\), it returns in the same round at
every agent. It returns $1$ exactly when at least one agent remains
unreached. Consequently, all agents either reject the current candidate and
proceed to the next trial in the same round or return the same value $U$ at
line~\ref{line:ub-return} in the same round.

Suppose that the agents return in the trial with parameter $q$.
Then \Flood returned $0$, so no agent was unreached and $R_q=V$.
Since $|R_q|\leq U$, it follows that $n=|V|=|R_q|\leq U$. Therefore, every
returned value is a valid upper bound.

It remains to prove termination. Consider
the trial with parameter $q=n-k$. Every agent has degree at most $q+k-1 = n-1$,
so every reached agent sends $1$ at line~\ref{line:ub-send}. If some agent
is still unreached at the beginning of a propagation round, connectivity of the
communication graph implies that there is an edge between a reached agent and an
unreached agent. The reached endpoint sends $1$, and hence at least one new agent
becomes reached at line~\ref{line:ub-update}. Starting with the leaders, all $n$
agents therefore become reached within $n-k$ rounds. The call to \Flood
returns $0$, and all agents terminate.

Hence \UpperBound terminates for every network size.
\end{proof}

\begin{lemma}\label{lem:bc-upper-bound-complexity}
Assume that \(1\leq k\leq n\). Then \UpperBound returns a value
\(U\leq n^{n-1}\) and uses \(O(n^n)\) rounds.
\end{lemma}

\begin{proof}
By the proof of Lemma~\ref{lem:bc-upper-bound-correct}, the trial with
parameter \(q=n-k\) succeeds. Therefore, the algorithm terminates in a trial
with some parameter \(q^\ast\leq n-k\). The returned value is
\(U=k(q^\ast+k)^{q^\ast}\).

Since \(k\leq n\), \(q^\ast+k\leq n\), and \(q^\ast\leq n-k\), we obtain
\(U\leq k n^{n-k}\leq n\cdot n^{n-2}=n^{n-1}\), if \(k\geq 2\).
If \(k=1\), then \(U=(q^\ast+1)^{q^\ast}\leq n^{n-1}\).
Thus, in all cases, \(U\leq n^{n-1}\).

A trial with parameter \(q\) uses exactly \(q\) rounds for the propagation
phase. By part~(3) of Lemma~\ref{lem:flood-correctness}, the call to \Flood
uses exactly \(k(q+k)^q\) rounds. Since the algorithm terminates with
\(q^\ast\leq n-k\), fewer than \(n\) trials are executed.

The propagation phases use fewer than \(n^2\) rounds in total. Moreover, by
the bound proved above, every flooding phase in an executed trial uses at most
\(n^{n-1}\) rounds. Since fewer than \(n\) trials are executed, the flooding
phases use at most \(n\cdot n^{n-1}=n^n\) rounds in total. Hence the overall
number of rounds is \(O(n^n)\).
\end{proof}

Already at this point, combining \UpperBound{} with the known-bound algorithms of
this section gives the following immediate corollary.

\begin{corollary}
In anonymous \(1\)-interval-connected dynamic networks with exactly
\(k\geq 1\) leaders, with \(k\) known to all agents and with a local degree
oracle, \ProblemName{Input Frequency}, \ProblemName{Input Multiset}, and \ProblemName{Counting} can be computed with
termination without any initially known upper bound on the network size.
The round complexities are $O(n^nB_{\max}+n^{n+1}\log n)$ for \ProblemName{Input Frequency} and \ProblemName{Input Multiset}, and $O(n^{n+1}\log n)$
for \ProblemName{Counting}.
\end{corollary}

\begin{proof}
Run \UpperBound{} first. By Lemma~\ref{lem:bc-upper-bound-correct}, all
agents terminate with the same value $U$ satisfying $U\geq n$. This
value can then be used as the common upper bound required by
Theorem~\ref{thm:anonymous-bc-frequency},
Theorem~\ref{thm:anonymous-bc-input-multiset}, and
Corollary~\ref{cor:anonymous-bc-classical-counting}. Lemma~\ref{lem:bc-upper-bound-complexity}
gives \(U\leq n^{n-1}\) and \(O(n^n)\) rounds to compute the bound. Substituting
this \(U\) into the known-bound complexities gives
$
O(n^{n-1}\cdot nB_{\max}+n^{n-1}\cdot n^2\log n)
=O(n^nB_{\max}+n^{n+1}\log n)
$
for \ProblemName{Input Frequency} and \ProblemName{Input Multiset}, and
$
O(n^{n-1}\cdot n^2\log n)=O(n^{n+1}\log n)
$
for \ProblemName{Counting}. These terms dominate the cost of \UpperBound{}.
\end{proof}

\section{Adaptive Flooding Without a Known Bound}
\label{sec:anonymous-bc-unknown-bound}

The known-bound algorithms of the previous sections use
calls of the form \(\Flood(U-1,x)\), where all agents know a common upper
bound \(U\geq n\).  Such a call computes the OR of the local input bits and,
when all agents start it in the same round, also gives all agents a common
finishing round.

We now remove the initially known bound by providing a communication layer that only requires a leader and can be used to simulate the \ProblemName{Input Multiset}
algorithm described in the previous section, leading to a stabilizing algorithm.

The communication layer maintains an estimate \(\widehat U\) of
the network size and runs the known-bound \ProblemName{Input Multiset} algorithm in successive
\emph{trials}.  A trial starts either at the beginning of the execution or when
the layer delivers a \Restart{} event.  When a \Restart{} event is delivered,
the current copy of the simulated algorithm is discarded and a fresh copy is
started with the current value of \(\widehat U\).

The communication layer exposes to the simulation algorithm the following interface.
\begin{itemize}
    \item \(\widehat U\) is the current estimate.  The simulated algorithm may
    read this value but cannot modify it.
    \item \AdaptiveFlooding\((x)\), where \(x\in\{0,1\}\), replaces every call
    to \(\Flood(U-1,x)\) of the known-bound algorithm.
    \item \DetectError{} is an operation by which the simulated algorithm may
    invalidate the current trial.  In the \ProblemName{Input Multiset} simulation below this
    happens when the computed value is not consistent with the estimate \(\widehat U\).
    \item \Restart{} is a local event delivered by the layer.  When it is
    delivered, the current trial is invalidated and the next trial starts from the
    initial state of the simulated algorithm.
\end{itemize}
Local events delivered by the layer are processed at the beginning of a round,
before the simulated algorithm takes its next step.

Before the final successful trial, trials are speculative: a call to
\AdaptiveFlooding{} may return any value in \(\{0,1\}\), and any output produced
by the simulated algorithm may be arbitrary.  Correctness is required only after
the last restart, or from the beginning if no restart occurs.

\newcommand{\IsInvalid}{\procname{IsInvalid}\xspace}

\subsection{Adaptive Flooding}

The \AdaptiveFlooding{} routine is given in Algorithm~\ref{alg:bc-adaptive-flooding}.

Each agent stores the estimate \(\widehat U\), initially \(1\).  It also
stores a boolean flag \(\mathit{invalid}\), initially false.  An agent with
\(\mathit{invalid}=\boolname{true}\) has abandoned the current trial; an agent
with \(\mathit{invalid}=\boolname{false}\) is \emph{active}.  The operation
\DetectError{} invalidates the current trial at the executing agent by setting
\(\mathit{invalid}\) to true.  The simulated algorithm may read \(\widehat U\)
and may call \DetectError{}, but it cannot read or modify
\(\mathit{invalid}\).  The read-only predicate \IsInvalid{}, used later by the
terminating wrapper, is the predicate \(\mathit{invalid}=\boolname{true}\).

A call to \AdaptiveFlooding{} has a \emph{common start} if all active agents
invoke the call in the same round.  It terminates \emph{synchronously} if all
active agents receive its return value in the same round.

The routine \AdaptiveFlooding{} uses \(\widehat U\) in place of the unknown
upper bound.  It first floods the input bit for \(\widehat U-1\) rounds and
obtains a local value \(y\).

Then two floodings are executed in which the leader certifies the value it obtained as the result of the first flood, and other agents verify that the value each received is consistent with the one received by the leader.  In
the first one, the leader's initial bit is \(1\) exactly when the leader's value
is \(0\); in the second one, the leader's initial bit is \(1\) exactly when the
leader's value is \(1\).  Thus the leader always participates in both
certificate floodings, but its pair of initial certificate bits is either
\((1,0)\) or \((0,1)\), never \((0,0)\). All other agents participate in the two floodings with $(0,0)$. An agent accepts its local value only
if it receives the certificate corresponding to that value.

\begin{algorithm}[H]
\caption{Adaptive flooding}
\label{alg:bc-adaptive-flooding}

\KwInput{A local boolean value $x$}
\KwOutput{A boolean value}

\Fn{\AlgAdaptiveFlooding{$x$}}{
  $y\leftarrow\Flood(\widehat U-1,x)$\label{line:bc-af-value-flood}\;

  $c_0\leftarrow
    \Flood(\widehat U-1,\mathit{isLeader}\land y=0)$\label{line:bc-af-cert-zero}\;

  $c_1\leftarrow
    \Flood(\widehat U-1,\mathit{isLeader}\land y=1)$\label{line:bc-af-cert-one}\;

  \If{$(y=0\land\lnot c_0)$ or $(y=1\land\lnot c_1)$\label{line:bc-af-check}}{
    \AlgDetectError{}\label{line:bc-af-detect}\;
  }

  \Return{$y$}\label{line:bc-af-return}\;
}
\end{algorithm}

The two certificates are needed because the absence of a one-bit signal is not
informative by itself: it may mean that the leader certified the other value, or
that the matching certificate did not reach the agent.  With two certificates,
every agent that returns without calling \DetectError{} can verify that its
local value agrees with the leader's value.

We now show that if the estimate $\widehat{U} \geq n$ then \AdaptiveFlooding behaves like a normal flood and no agent flags the current trial as invalid.

\begin{lemma}
\label{lem:bc-adaptive-flooding-correctness}
Assume that no agent is invalid before a common-start call to
\AdaptiveFlooding{}, and that all agents use the same estimate
\(\widehat U\) in this call.  Then all agents reach the return statement in
the same round.  Moreover, either some agent becomes invalid during the call,
or all agents return the global OR of their input bits.  If
\(\widehat U\geq n\), then all agents return the global OR and no agent
becomes invalid.
\end{lemma}

\begin{proof}
Consider a common-start call and assume that no agent is invalid before it
starts.  Since all agents use the same estimate, the three floodings in
lines~\ref{line:bc-af-value-flood}--\ref{line:bc-af-cert-one} have the same
fixed durations at all agents.  Hence all agents reach the check in
line~\ref{line:bc-af-check} and the return statement in
line~\ref{line:bc-af-return} in the same round.

Let \(b\) be the leader's value after the first flooding.  In the certificate
flooding of line~\ref{line:bc-af-cert-zero}, the leader's initial bit is \(1\)
exactly if \(b=0\); in the certificate flooding of
line~\ref{line:bc-af-cert-one}, the leader's initial bit is \(1\) exactly if
\(b=1\).  No other agent can have initial bit \(1\) in either certificate
flooding.

Suppose that an agent reaches line~\ref{line:bc-af-return} without calling
\DetectError{} in line~\ref{line:bc-af-detect}.  If its local value is \(y=0\),
then the test in line~\ref{line:bc-af-check} implies that it received
\(c_0=1\).  By Lemma~\ref{lem:flood-correctness}, this is possible only if some
agent had initial bit \(1\) in the certificate-for-\(0\) flooding of
line~\ref{line:bc-af-cert-zero}; hence the leader had \(b=0\).  Similarly, if
its local value is \(y=1\), then receiving \(c_1=1\) implies, through
line~\ref{line:bc-af-cert-one}, that the leader had \(b=1\).  Therefore every
agent that returns without invalidating the trial has \(y=b\).  If no agent
becomes invalid during the call, all agents return the same value \(b\).

If some input bit is \(1\), then an agent holding such a bit has value \(1\)
after the value-flooding of line~\ref{line:bc-af-value-flood}, because flooding
never changes a \(1\) to \(0\).  Since all returned values are equal in the
non-invalid case, the common value must be \(1\).  If all input bits are \(0\),
flooding cannot create a \(1\), by Lemma~\ref{lem:flood-correctness}; hence the
common value is \(0\).  Thus, in the non-invalid case, the returned value is
exactly the global OR.

Finally suppose that \(\widehat U\geq n\).  Each internal flooding in
lines~\ref{line:bc-af-value-flood}--\ref{line:bc-af-cert-one} has duration
\(\widehat U-1\geq n-1\).  By Corollary~\ref{cor:flood-upper-bound}, the first
flooding returns the global OR at every agent, and the certificate for this
common value reaches every agent.  Therefore the condition in
line~\ref{line:bc-af-check} is false at every agent, and no agent calls
\DetectError{}.
\end{proof}

\subsection{Recovery Service}

The purpose of recovery is to turn a call to \DetectError{} at one agent into
a synchronized restart with a larger estimate.  The implementation uses three
logical recovery protocols, called invalidation, control, and reset.
These protocols are multiplexed by interleaving the rounds of their execution.

An execution of a recovery protocol means one round of that logical protocol.
Consecutive executions are consecutive inside that protocol's own subsequence;
after the protocols are interleaved, they need not be adjacent rounds of the
underlying network.

For an estimate \(u\), a \emph{control window} is a group
of \(u\) consecutive executions of the control protocol, measured from the
beginning of the current trial.  At the first execution of a control window, an
active leader starts a control signal, and active agents relay it until the
window ends.

A \emph{reset window} is defined analogously for the reset
protocol: it is a group of consecutive executions of the reset protocol, its
length is prescribed by a deterministic reset calendar known to all agents,
and reached agents update their estimate only at the common end of the
window.

The recovery service uses the invalidation, control, and reset protocols to achieve the synchronized restart in three phases.

\paragraph{Invalidation propagation.} The purpose of this phase is to propagate the information of an invalid trial to the entire network.
Once an agent becomes invalid, it repeatedly announces this fact.  Any agent
that hears such an announcement also calls \DetectError{} and becomes invalid.
Therefore, if a trial is invalidated anywhere, invalidation eventually reaches
the leader.

\paragraph{Silent control window.}
The leader does not reset immediately after becoming invalid.  Instead, it waits
through one complete control window of length equal to the current estimate.  In
a normal control window, an active leader would start the control signal.  An
invalid leader does not send such a signal.  Hence every agent that still trusts the
old trial misses the signal and calls \DetectError{} by the end of the window.
Thus, before the reset starts, no agent still trusts the old trial.

\paragraph{Calendar reset.}
After the silent control window has ended, the leader waits for a reset window
of length \(2\widehat U\) and floods a reset signal through that window.  Every
agent reached by the signal installs the estimate \(2\widehat U\), propagates the signal in the window, and when the window ends, clears its
invalid flag and generates a \Restart{}. If the window is too short to reach all agents, the
unreached agents remain invalid; their invalidation announcements eventually
force another reset.  If \(2\widehat U\geq n\), the reset reaches every agent.

The reset calendar is local and deterministic.  It is divided into stages.
Stage \(h\) is the finite sequence of reset windows whose lengths are
\(1,2,4,\ldots,2^h\).  Thus the calendar begins
\[
    1;\qquad 1,2;\qquad 1,2,4;\qquad 1,2,4,8;\qquad\ldots
\]
and every power of two appears infinitely often.  No messages are needed to
construct the calendar: all agents advance the same deterministic calendar
whenever the reset protocol is scheduled, so they agree on the length and the
end of the current reset window.

Figure~\ref{fig:bc-reset-calendar-interleaving} illustrates the interleaving of the channels and the concept of a reset window.

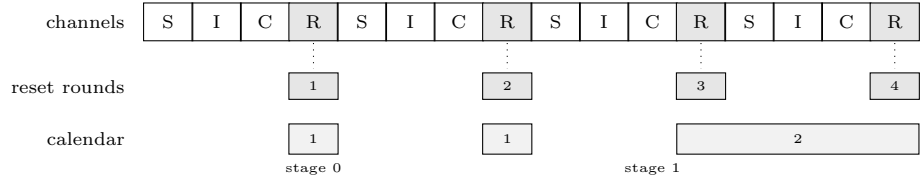
\begin{figure}[t]
\centering
\begin{tikzpicture}[
    slot/.style={draw,minimum width=6.5mm,minimum height=5mm,inner sep=0pt,font=\scriptsize},
    rslot/.style={slot,fill=black!10},
    tick/.style={draw,fill=black!10,minimum width=6.5mm,minimum height=3.5mm,inner sep=0pt,font=\tiny},
    win/.style={draw,fill=black!5,minimum height=4mm,inner sep=1pt,font=\tiny},
    lab/.style={font=\scriptsize,anchor=east},
    note/.style={font=\tiny}
]
\node[lab] at (-.45,0) {channels};
\foreach \i/\c/\sty in {
  0/S/slot,1/I/slot,2/C/slot,3/R/rslot,
  4/S/slot,5/I/slot,6/C/slot,7/R/rslot,
  8/S/slot,9/I/slot,10/C/slot,11/R/rslot,
  12/S/slot,13/I/slot,14/C/slot,15/R/rslot} {
  \pgfmathsetmacro{\x}{0.64*\i}
  \node[\sty] at (\x,0) {\c};
}

\node[lab] at (-.45,-.85) {reset rounds};
\foreach \x/\j in {1.92/1,4.48/2,7.04/3,9.60/4} {
  \node[tick] (rr\j) at (\x,-.85) {\j};
  \draw[dotted] (\x,-.28)--(\x,-.62);
}

\node[lab] at (-.45,-1.55) {calendar};
\node[win,minimum width=6.5mm] at (1.92,-1.55) {1};
\node[win,minimum width=6.5mm] at (4.48,-1.55) {1};
\node[win,minimum width=32mm] at (8.32,-1.55) {2};
\node[note] at (1.92,-1.95) {stage \(0\)};
\node[note] at (6.40,-1.95) {stage \(1\)};
\end{tikzpicture}
\caption{The four channels are interleaved in the fixed order: simulation (S),
 invalidation propagation (I), control (C), reset (R).  The reset calendar advances only
on the reset-channel subsequence.  Thus the window labeled \(2\) consists of two
reset-channel rounds, shown as reset rounds \(3\) and \(4\).}
\label{fig:bc-reset-calendar-interleaving}
\end{figure}

\subsection{Implementation by Four Interleaved Channels}

A \emph{channel} is a logical protocol executed on a prescribed subsequence of
rounds.  The implementation uses four channels: simulation, invalidation,
control, and reset.  The scheduler assigns rounds cyclically to these channels
in that order.  A \emph{channel round} is one execution of the corresponding
channel.  Since the network graph is connected in every round, the graph
sequence observed by each channel is connected in every channel round.

The simulation channel executes the simulated algorithm.  Each round of an
unfinished \AdaptiveFlooding{} call is executed as one simulation-channel round.
At the beginning of every simulation-channel round, the layer first checks
whether a restart flag has been set.  If so, it delivers \Restart{} to the
simulated algorithm and aborts any unfinished simulated operation.

\smallskip

The recovery service is implemented by Algorithms~\ref{alg:bc-channel-interleaving}, \ref{alg:bc-error-control-rounds} and \ref{alg:bc-reset-rounds}.

\paragraph{Recovery service initialization.}
Algorithm~\ref{alg:bc-channel-interleaving}
initializes the estimate, flags, control-position, and reset-calendar variables in
lines~\ref{line:bc-init-estimate}--\ref{line:bc-init-reset-countdown}.  It then
runs the four logical channels cyclically: the simulation channel is executed in
lines~\ref{line:bc-run-read-restart}--\ref{line:bc-run-simulation}, followed by
one invalidation, control, and reset channel round in
lines~\ref{line:bc-run-invalidation}--\ref{line:bc-run-reset}.

\paragraph{Invalidation and control.}

Algorithm~\ref{alg:bc-error-control-rounds} contains the invalidation and
control logic.  The operation \DetectError{} marks the current trial as invalid
in lines~\ref{line:bc-de-test-active}--\ref{line:bc-de-invalid}; if the leader
is the agent that becomes invalid, lines~\ref{line:bc-de-test-leader}--\ref{line:bc-de-reset-ready-false}
start the wait for a complete silent control window.  Invalidation is propagated
by lines~\ref{line:bc-inv-send}--\ref{line:bc-inv-detect} of Algorithm~\ref{alg:bc-error-control-rounds}.

The control window is also implemented directly in Algorithm~\ref{alg:bc-error-control-rounds}.
The local variable \(\mathit{controlPos}\), initialized in
line~\ref{line:bc-init-control-pos} of Algorithm~\ref{alg:bc-channel-interleaving}
and updated in lines~\ref{line:bc-control-pos-test}--\ref{line:bc-control-pos-inc}
of Algorithm~\ref{alg:bc-error-control-rounds}, records the current position in
the control window.  When the current estimate is \(\widehat U\), this position ranges
over \(1,\ldots,\widehat U\); position \(1\) is the first control-channel round of
the window and position \(\widehat U\) is its last round.  When an agent installs a
new estimate, line~\ref{line:bc-reset-control-pos-zero} of
Algorithm~\ref{alg:bc-reset-rounds} resets \(\mathit{controlPos}\) to \(0\), so the
next control-channel round starts a fresh control window for the new trial.  The
leader uses the flags \(\mathit{waitSilent}\) and \(\mathit{resetReady}\): the first
means that the leader has become invalid and is waiting for one complete silent
control window, while the second means that this silent window has ended and the
leader may start the next reset window of length \(2\widehat U\).

The control position is
advanced in lines~\ref{line:bc-control-pos-test}--\ref{line:bc-control-pos-inc};
the beginning of a control window is handled in
lines~\ref{line:bc-control-window-start}--\ref{line:bc-control-silent-set}; the
control signal is sent and relayed in lines~\ref{line:bc-control-send}--\ref{line:bc-control-relay};
and a missing signal or the end of the silent window is handled in
lines~\ref{line:bc-control-missing}--\ref{line:bc-control-silent-clear}.

\paragraph{Reset and restart.}
Algorithm~\ref{alg:bc-reset-rounds} contains the reset calendar and reset
flooding.  Lines~\ref{line:bc-cal-stage}--\ref{line:bc-cal-yield} of Algorithm~\ref{alg:bc-reset-rounds} define the
calendar.  Lines~\ref{line:bc-reset-new-window-test}--\ref{line:bc-reset-countdown-init}
open a new reset window when the previous one ends.  Once the silent control
window has ended, the leader starts a reset exactly at the beginning of the next
reset window of length \(2\widehat U\), by
lines~\ref{line:bc-reset-start-test}--\ref{line:bc-reset-clear-ready}.  The
reset signal is flooded in lines~\ref{line:bc-reset-send}--\ref{line:bc-reset-relay}.
Agents reached by the reset signal wait until the common end of that window,
then lines~\ref{line:bc-reset-end-test}--\ref{line:bc-reset-set-restart} install
the window value as the new estimate, clear the invalid flag, reset the control
position, and set the local flag that will be delivered as \Restart{} on the next
simulation-channel round.

\begin{algorithm}[H]
\SetKwFunction{AlgInvalidationRound}{InvalidationRound}
\caption{Initializing and scheduling the four-channel layer}
\label{alg:bc-channel-interleaving}

\Fn{\AlgInitializeCommunication{}}{
  $\widehat U\leftarrow1$\label{line:bc-init-estimate}\;
  $\mathit{invalid}\leftarrow\boolname{false}$\label{line:bc-init-invalid}\;
  $\mathit{restart}\leftarrow\boolname{false}$\label{line:bc-init-restart}\;
  $\mathit{reset}\leftarrow\boolname{false}$\label{line:bc-init-reset}\;
  $\mathit{control}\leftarrow\boolname{false}$\label{line:bc-init-control}\;
  $\mathit{controlPos}\leftarrow0$\label{line:bc-init-control-pos}\;
  \tcp{position inside the current control window; 0 means no window has started}
  $\mathit{waitSilent}\leftarrow\boolname{false}$\label{line:bc-init-wait-silent}\;
  \tcp{leader flag: waiting for one complete silent control window}
  $\mathit{silentWindow}\leftarrow\boolname{false}$\label{line:bc-init-silent-window}\;
  $\mathit{resetReady}\leftarrow\boolname{false}$\label{line:bc-init-reset-ready}\;
  \tcp{leader flag: the silent control window has ended}
  $U_{\mathit{reset}}\leftarrow\bot$\label{line:bc-init-reset-value}\;
  $T_{\mathit{reset}}\leftarrow0$\label{line:bc-init-reset-countdown}\;
}

\Fn{\AlgRunAdaptiveLayer{}}{
  \AlgInitializeCommunication{}\label{line:bc-run-init}\;

  \While{\boolname{true}}{
    $\mathit{restartEvent}\leftarrow\mathit{restart}$\label{line:bc-run-read-restart}\;
    $\mathit{restart}\leftarrow\boolname{false}$\label{line:bc-run-clear-restart}\;
    execute one simulation-channel round, delivering a \Restart{} event iff
    $\mathit{restartEvent}$\label{line:bc-run-simulation}\;

    \AlgInvalidationRound{}\label{line:bc-run-invalidation}\;

    \AlgControlRound{}\label{line:bc-run-control}\;

    \AlgResetRound{}\label{line:bc-run-reset}\;
  }
}
\end{algorithm}

\begin{algorithm}[H]
\SetKwFunction{AlgInvalidationRound}{InvalidationRound}
\caption{Invalidation propagation and control windows}
\label{alg:bc-error-control-rounds}

\Fn{\AlgDetectError{}}{
  \If{$\lnot\mathit{invalid}$\label{line:bc-de-test-active}}{
    $\mathit{invalid}\leftarrow\boolname{true}$\label{line:bc-de-invalid}\;

    \If{$\mathit{isLeader}$\label{line:bc-de-test-leader}}{
      $\mathit{waitSilent}\leftarrow\boolname{true}$\label{line:bc-de-wait-silent}\;
      $\mathit{silentWindow}\leftarrow\boolname{false}$\label{line:bc-de-silent-false}\;
      $\mathit{resetReady}\leftarrow\boolname{false}$\label{line:bc-de-reset-ready-false}\;
    }
  }
}

\Fn{\AlgInvalidationRound{}}{
  send $1$ iff $\mathit{invalid}$; otherwise send $0$ on the invalidation channel\label{line:bc-inv-send}\;

  \If{at least one neighbor sent $1$\label{line:bc-inv-receive}}{
    \AlgDetectError{}\label{line:bc-inv-detect}\;
  }
}

\Fn{\AlgControlRound{}}{
  \tcp{advance the position inside the current control window}
  \eIf{$\mathit{controlPos}=0$ or $\mathit{controlPos}=\widehat U$\label{line:bc-control-pos-test}}{
    $\mathit{controlPos}\leftarrow1$\label{line:bc-control-pos-reset}\;
  }{
    $\mathit{controlPos}\leftarrow\mathit{controlPos}+1$\label{line:bc-control-pos-inc}\;
  }

  \tcp{beginning of a control window}
  \If{$\mathit{controlPos}=1$\label{line:bc-control-window-start}}{
    $\mathit{control}\leftarrow\mathit{isLeader}\land\lnot\mathit{invalid}$\label{line:bc-control-start-signal}\;
    $\mathit{silentWindow}\leftarrow
      \mathit{isLeader}\land\mathit{invalid}\land\mathit{waitSilent}$\label{line:bc-control-silent-set}\;
  }

  send $1$ iff $\mathit{control}\land\lnot\mathit{invalid}$;
  otherwise send $0$ on the control channel\label{line:bc-control-send}\;

  \If{$\lnot\mathit{invalid}$ and at least one neighbor sent $1$\label{line:bc-control-receive}}{
    $\mathit{control}\leftarrow\boolname{true}$\label{line:bc-control-relay}\;
  }

  \tcp{end of a control window: active agents must have received the control signal}
  \If{$\mathit{controlPos}=\widehat U$ and
      $\lnot\mathit{invalid}$ and $\lnot\mathit{control}$\label{line:bc-control-missing}}{
    \AlgDetectError{}\label{line:bc-control-detect}\;
  }

  \If{$\mathit{controlPos}=\widehat U$ and $\mathit{silentWindow}$\label{line:bc-control-silent-end}}{
    $\mathit{resetReady}\leftarrow\boolname{true}$\label{line:bc-control-ready}\;
    $\mathit{waitSilent}\leftarrow\boolname{false}$\label{line:bc-control-wait-clear}\;
    $\mathit{silentWindow}\leftarrow\boolname{false}$\label{line:bc-control-silent-clear}\;
  }
}
\end{algorithm}

\begin{algorithm}[H]
\caption{Reset calendar and reset windows}
\label{alg:bc-reset-rounds}

\Fn{\AlgResetCalendar{}}{
  \For{$h\leftarrow0,1,2,\ldots$\label{line:bc-cal-stage}}{
    \For{$i\leftarrow0$ \KwTo $h$\label{line:bc-cal-index}}{
      yield $2^i$\label{line:bc-cal-yield}\;
    }
  }
}

\Fn{\AlgResetRound{}}{
  \If{$T_{\mathit{reset}}=0$\label{line:bc-reset-new-window-test}}{
    $U_{\mathit{reset}}\leftarrow\AlgResetCalendar{}$\label{line:bc-reset-calendar-next}\;
    $T_{\mathit{reset}}\leftarrow U_{\mathit{reset}}$\label{line:bc-reset-countdown-init}\;
  }

  \If{$T_{\mathit{reset}}=U_{\mathit{reset}}$ and
      $U_{\mathit{reset}}=2\widehat U$ and
      $\mathit{isLeader}$ and $\mathit{resetReady}$\label{line:bc-reset-start-test}}{
    $\mathit{reset}\leftarrow\boolname{true}$\label{line:bc-reset-start}\;
    $\mathit{resetReady}\leftarrow\boolname{false}$\label{line:bc-reset-clear-ready}\;
  }

  send $1$ iff $\mathit{reset}$; otherwise send $0$ on the reset channel\label{line:bc-reset-send}\;

  \If{at least one neighbor sent $1$\label{line:bc-reset-receive}}{
    $\mathit{reset}\leftarrow\boolname{true}$\label{line:bc-reset-relay}\;
  }

  \If{$T_{\mathit{reset}}=1$ and $\mathit{reset}$\label{line:bc-reset-end-test}}{
    $\mathit{invalid}\leftarrow\boolname{false}$\label{line:bc-reset-clear-invalid}\;
    $\mathit{reset}\leftarrow\boolname{false}$\label{line:bc-reset-clear-reset}\;
    $\mathit{control}\leftarrow\boolname{false}$\label{line:bc-reset-clear-control}\;
    $\mathit{controlPos}\leftarrow0$\label{line:bc-reset-control-pos-zero}\;
    $\mathit{waitSilent}\leftarrow\boolname{false}$\label{line:bc-reset-clear-wait}\;
    $\mathit{silentWindow}\leftarrow\boolname{false}$\label{line:bc-reset-clear-silent}\;
    $\mathit{resetReady}\leftarrow\boolname{false}$\label{line:bc-reset-clear-ready-final}\;
    $\widehat U\leftarrow U_{\mathit{reset}}$\label{line:bc-reset-update-estimate}\;
    \tcp{restart the simulation and start a new trial}
    $\mathit{restart}\leftarrow\boolname{true}$\label{line:bc-reset-set-restart}\;
  }

  $T_{\mathit{reset}}\leftarrow T_{\mathit{reset}}-1$\label{line:bc-reset-decrement}\;
}
\end{algorithm}

\subsection{Correctness and Cost of Recovery}

We first record the synchronization properties that follow from the deterministic
scheduler and reset calendar.

\begin{lemma}
\label{lem:bc-communication-synchronization}
All agents execute the same channel in the same rounds.  Moreover, all
agents advance the reset calendar synchronously and therefore agree on the
boundaries and lengths of all reset windows.
\end{lemma}

\begin{proof}
The four-channel loop in Algorithm~\ref{alg:bc-channel-interleaving} is fixed:
one simulation-channel round is executed in line~\ref{line:bc-run-simulation},
then the invalidation, control, and reset channels are executed in
lines~\ref{line:bc-run-invalidation}, \ref{line:bc-run-control}, and
\ref{line:bc-run-reset}.  Since every agent runs this same loop, all agents
execute the same channel in the same rounds.

It remains to consider the reset calendar.  Initially all agents set
\(U_{\mathit{reset}}=\bot\) and \(T_{\mathit{reset}}=0\) in
lines~\ref{line:bc-init-reset-value}--\ref{line:bc-init-reset-countdown}, and no
agent has requested a calendar value.  Suppose these values, and the number of
previous calendar requests, are the same at all agents at the beginning of a
reset-channel round.  The test in line~\ref{line:bc-reset-new-window-test} has
the same truth value everywhere.  If it is true, all agents request the next
calendar value in line~\ref{line:bc-reset-calendar-next}; the calendar generated
by lines~\ref{line:bc-cal-stage}--\ref{line:bc-cal-yield} of Algorithm~\ref{alg:bc-reset-rounds} is deterministic, so
they receive the same value and set the same countdown in
line~\ref{line:bc-reset-countdown-init}.  If the test is false, no agent
requests a new value.  In both cases all agents decrement the countdown once
in line~\ref{line:bc-reset-decrement}.  Thus the calendar state remains
synchronized by induction, and the reset-window boundaries and lengths are
common to all agents.
\end{proof}

\begin{lemma}
\label{lem:bc-channel-overhead}
An execution of \(q\) rounds on any one of the four channels uses at most
\(4q\) rounds.
\end{lemma}

\begin{proof}
The scheduler in Algorithm~\ref{alg:bc-channel-interleaving} executes the
channels cyclically: one simulation-channel round in line~\ref{line:bc-run-simulation},
one invalidation-channel round in line~\ref{line:bc-run-invalidation}, one
control-channel round in line~\ref{line:bc-run-control}, and one reset-channel
round in line~\ref{line:bc-run-reset}.  Hence each channel is selected exactly
once in every four consecutive rounds.  The interval from the first to the
\(q\)-th occurrence of a fixed channel has length \(4(q-1)+1\leq4q\).
\end{proof}

\begin{lemma}
\label{lem:bc-error-propagation}
If at least one agent is invalid, then every agent becomes invalid within
\(n-1\) invalidation-channel rounds, and hence within \(4(n-1)\) rounds.
\end{lemma}

\begin{proof}
In an invalidation-channel round, invalid agents send \(1\) in
line~\ref{line:bc-inv-send}.  Any active agent that receives such a \(1\) calls
\DetectError{} in line~\ref{line:bc-inv-detect} and becomes invalid in
line~\ref{line:bc-de-invalid}.  As long as some agent is active and some
agent is invalid, connectedness provides an edge crossing the cut between the
two sets in the next invalidation-channel round.  Therefore at least one active
agent becomes invalid in each such round.  Starting from one invalid agent,
after at most \(n-1\) invalidation-channel rounds all agents are invalid. The
bound in rounds follows from Lemma~\ref{lem:bc-channel-overhead}.
\end{proof}

\begin{lemma}
\label{lem:bc-errors-trigger-resets}
\label{lem:bc-error-reset-correctness}
\label{lem:bc-common-estimate}
Throughout the execution, all active agents have the same estimate and the
same control position.  Fix an interval between two consecutive reset completions,
or the initial interval before the first reset completion, and let \(u\) be the
common estimate of the active agents at the beginning of the interval.  If
some agent is invalid during this interval, then the leader eventually
initiates a reset flooding at the beginning of a reset window of length \(2u\).
At the beginning of that window all agents are invalid.  All agents
reached during the window set their estimate to \(2u\), reset their control
position to \(0\), and observe \Restart{} in the same round.  If \(2u\geq n\), the
reset reaches all agents.
\end{lemma}

\begin{proof}

We argue by induction over reset completions.  Initially all agents are
active, have estimate \(1\), and have control position \(0\), by
lines~\ref{line:bc-init-estimate}, \ref{line:bc-init-invalid}, and
\ref{line:bc-init-control-pos} of Algorithm~\ref{alg:bc-channel-interleaving}.  Thus the invariant holds before the first
reset completion.  Assume that, at the beginning of the current interval, all
active agents have the same estimate \(u\) and the same control position.  While
no reset completes, estimates do not change.  In every control-channel round,
all active agents update their position by the same deterministic rule in
lines~\ref{line:bc-control-pos-test}--\ref{line:bc-control-pos-inc} of Algorithm~\ref{alg:bc-error-control-rounds}, using
the same estimate \(u\).  If an active agent becomes invalid, it leaves the
active set by line~\ref{line:bc-de-invalid} of Algorithm~\ref{alg:bc-error-control-rounds}.  Hence all active agents
continue to have the same estimate and the same control position throughout the
interval.

If some agent is invalid, Lemma~\ref{lem:bc-error-propagation} implies that
invalidation reaches every agent within at most \(n-1\) invalidation-channel
rounds; in particular, the leader eventually becomes invalid.

When the leader first becomes invalid in this interval, it has estimate \(u\). It
sets \(\mathit{waitSilent}\) to true and clears the two waiting flags in
lines~\ref{line:bc-de-wait-silent}--\ref{line:bc-de-reset-ready-false} of Algorithm~\ref{alg:bc-error-control-rounds}.  Consider
the first control window whose first round occurs after this event.  At the first
round of this window, all active agents enter position \(1\) together by
lines~\ref{line:bc-control-pos-test}--\ref{line:bc-control-pos-reset} of Algorithm~\ref{alg:bc-error-control-rounds}.  The
leader is invalid and \(\mathit{waitSilent}=\boolname{true}\), so the leader
sets \(\mathit{silentWindow}\) to true in line~\ref{line:bc-control-silent-set} of Algorithm~\ref{alg:bc-error-control-rounds}
and does not start a control signal in line~\ref{line:bc-control-start-signal} of Algorithm~\ref{alg:bc-error-control-rounds}.
Every active non-leader also sets its control flag to false in
line~\ref{line:bc-control-start-signal} of Algorithm~\ref{alg:bc-error-control-rounds}.  During the window, invalid agents
send \(0\) on the control channel by line~\ref{line:bc-control-send} of Algorithm~\ref{alg:bc-error-control-rounds}, and no
active agent can create a control signal unless it first receives one and
executes line~\ref{line:bc-control-relay} of Algorithm~\ref{alg:bc-error-control-rounds}.  Since the invalid leader starts no
signal, no active agent receives one in the window.  Therefore, at the last
round of the window, every agent that is still active satisfies the test in
line~\ref{line:bc-control-missing} of Algorithm~\ref{alg:bc-error-control-rounds} and calls \DetectError{} in
line~\ref{line:bc-control-detect} of Algorithm~\ref{alg:bc-error-control-rounds}.  Thus all agents are invalid by the end of
the silent control window.  At that same last round, the leader sets
\(\mathit{resetReady}\) to true in line~\ref{line:bc-control-ready} of Algorithm~\ref{alg:bc-error-control-rounds} and clears
the waiting flags in lines~\ref{line:bc-control-wait-clear}--\ref{line:bc-control-silent-clear} of Algorithm~\ref{alg:bc-error-control-rounds}.

After the silent control window, the leader waits until the first reset-calendar
window whose value is \(2u\).  Such a window exists because the calendar of
lines~\ref{line:bc-cal-stage}--\ref{line:bc-cal-yield} of Algorithm~\ref{alg:bc-reset-rounds} contains every power of
two infinitely often.  At the beginning of that window, the condition in
line~\ref{line:bc-reset-start-test} of Algorithm~\ref{alg:bc-reset-rounds} holds at the leader: the window has just
started, its value is \(2u\), and \(\mathit{resetReady}\) is true.  Hence the
leader sets \(\mathit{reset}\) to true in line~\ref{line:bc-reset-start} of Algorithm~\ref{alg:bc-reset-rounds} and
starts a reset flooding.  By Lemma~\ref{lem:bc-communication-synchronization},
all agents agree on the window boundaries and on the value
\(U_{\mathit{reset}}=2u\).  During the window, the reset signal is flooded by
lines~\ref{line:bc-reset-send}--\ref{line:bc-reset-relay}.  Every reached
agent has \(\mathit{reset}=\boolname{true}\) at the last reset-channel round
of the window and therefore satisfies the test in line~\ref{line:bc-reset-end-test}.
It then clears its invalid flag, resets its control position, installs the estimate
\(2u\), and sets its local restart flag in
lines~\ref{line:bc-reset-clear-invalid}--\ref{line:bc-reset-set-restart}.  Since
all reached agents agree on the end of the reset window, they perform these
updates in the same round and install the same new estimate and control position.

The restart flag is read and cleared in
lines~\ref{line:bc-run-read-restart}--\ref{line:bc-run-clear-restart}, and the
corresponding \Restart{} event is delivered in line~\ref{line:bc-run-simulation}.
By Lemma~\ref{lem:bc-communication-synchronization}, this is the same round for
all reached agents.  Hence they observe \Restart{} synchronously.

If \(2u\geq n\), then the reset window has length at least \(n\).  Starting from
the leader, the reset signal is flooded for at least \(n\) reset-channel rounds
by lines~\ref{line:bc-reset-send}--\ref{line:bc-reset-relay}.  By the flooding
argument of Lemma~\ref{lem:flood-correctness}, it reaches every agent by the
end of the window.

The only agents that become active after the reset are the reached agents,
which clear \(\mathit{invalid}\) in line~\ref{line:bc-reset-clear-invalid}; all
of them install the same estimate in line~\ref{line:bc-reset-update-estimate} and
reset their control position to \(0\) in line~\ref{line:bc-reset-control-pos-zero}.
Agents not reached, if any, remain invalid.  Thus the induction invariant is
preserved for the next interval.
\end{proof}

\begin{lemma}
\label{lem:bc-adaptive-flooding-cost}
One call to \AdaptiveFlooding{} with estimate \(\widehat U\) uses
\(O(\widehat U)\) rounds.
\end{lemma}

\begin{proof}
The routine performs three calls to \Flood{} in
lines~\ref{line:bc-af-value-flood}--\ref{line:bc-af-cert-one}, each of duration
\(\widehat U-1\).  Each round of these internal floodings is executed in one
simulation-channel round by line~\ref{line:bc-run-simulation}.  Thus the call
uses \(3(\widehat U-1)=O(\widehat U)\) simulation-channel rounds, and the bound
in rounds follows from Lemma~\ref{lem:bc-channel-overhead}.
\end{proof}

We separate two kinds of cost in the analysis.  The
\emph{calendar cost} is the passive delay after the leader is ready to reset but
before the deterministic reset calendar reaches a window of the required
length.  The \emph{non-calendar cost} of an interval with fixed leader estimate
is all other work charged to that interval: simulation work until some agent
aborts, invalidation propagation, the silent control window, and the actual flooding
inside the reset window once that window has begun.  Thus reset flooding is
non-calendar cost, while waiting for the appropriate reset-window boundary is
calendar cost.

\begin{lemma}[Non-calendar recovery cost]
\label{lem:bc-recovery-cost}
After the first invalidation in an interval with common active estimate \(u\),
the recovery cost excluding the waiting time for the appropriate reset-calendar
window is \(O(n+u)\) rounds.
\end{lemma}

\begin{proof}
By Lemma~\ref{lem:bc-error-propagation}, invalidation reaches all agents, and
in particular the leader, within \(O(n)\) rounds, using
lines~\ref{line:bc-inv-send}--\ref{line:bc-inv-detect} of Algorithm~\ref{alg:bc-error-control-rounds}.  When the leader becomes
invalid, lines~\ref{line:bc-de-wait-silent}--\ref{line:bc-de-reset-ready-false} of Algorithm~\ref{alg:bc-error-control-rounds}
start the wait for a complete silent control window.  The position update in
lines~\ref{line:bc-control-pos-test}--\ref{line:bc-control-pos-inc} of Algorithm~\ref{alg:bc-error-control-rounds} implies
that at most one partial control window remains before the first silent control
window starts, and that silent window has length \(u\).  At its end, the leader
sets \(\mathit{resetReady}\) in line~\ref{line:bc-control-ready} of Algorithm~\ref{alg:bc-error-control-rounds}.  Hence reaching
the end of the silent control window costs \(O(u)\) channel rounds, and therefore
\(O(u)\) rounds by Lemma~\ref{lem:bc-channel-overhead}.  When the reset window of
length \(2u\) begins, the leader starts the reset in
lines~\ref{line:bc-reset-start-test}--\ref{line:bc-reset-start} of Algorithm~\ref{alg:bc-reset-rounds}, and the reset
flooding of lines~\ref{line:bc-reset-send}--\ref{line:bc-reset-relay} of Algorithm~\ref{alg:bc-reset-rounds} costs
\(O(u)\) more rounds.  These are the only recovery costs not counted as calendar
waiting.
\end{proof}

\begin{lemma}[Reset-calendar waiting time]
\label{lem:bc-reset-calendar-complexity}
Suppose that \(t\) reset-channel rounds have elapsed when the leader starts
waiting for a reset window of length \(u=2^k\).  The desired window begins after
at most \(t+O(u+\log(t+2))\) additional reset-channel rounds.  Consequently, if \(T\) rounds have
elapsed at that time, then the waiting costs at most \(T+O(u+\log(T+2))\)
additional rounds.
\end{lemma}

\begin{proof}
The calendar is generated by \AlgResetCalendar in Algorithm \ref{alg:bc-reset-rounds}:
stage \(j\) consists of the reset windows of lengths \(1,2,4,\ldots,2^j\), and
therefore has total length \(2^{j+1}-1\).  Let the leader start waiting during
stage \(h\).  The total length of the stages preceding stage \(h\) is
\(2^{h+1}-h-2\), and hence \(h=O(\log(t+2))\).  Finishing the current stage costs at
most \(t+O(\log(t+2))\) additional reset-channel rounds.

After the current stage ends, a window of length \(u=2^k\) appears no later than
stage \(\max\{h+1,k\}\), again by the definition of \AlgResetCalendar.
If \(k\leq h+1\), then after finishing the current stage, the desired window
appears during stage \(h+1\), and the part of that stage before the window has
length less than \(u\).  If \(k>h+1\), then the total length of stages
\(h+1,\ldots,k\) before the desired window is \(O(2^k)=O(u)\).  Thus the desired
window begins after at most \(t+O(u+\log(t+2))\) additional reset-channel rounds.

After \(T\) ordinary rounds, at most \(\lfloor T/4\rfloor+1\) reset-channel
rounds have elapsed, and each additional reset-channel round costs at most four
ordinary rounds by the scheduler lines~\ref{line:bc-run-simulation}--\ref{line:bc-run-reset}.  The
waiting time in ordinary rounds is therefore at most \(T+O(u+\log(T+2))\).
\end{proof}

\begin{lemma}[Calendar amortization]
\label{lem:bc-calendar-amortization}
Let \(u_i=2^i\), and suppose that the leader initiates exactly \(r\) resets.
Let \(T_i\) be the elapsed time at the beginning of the interval with leader
estimate \(u_i\). If \(H_i\) is an upper bound on the non-calendar cost charged
to that interval, including simulation work before invalidation propagation,
the silent control window, and the reset flooding, and if the
sequence \(H_i\) is nondecreasing and satisfies \(H_i=\Omega(u_i)\), then the
total time until the beginning of the interval with leader estimate \(u_r\) is
\(O\left(\sum_{i=0}^{r-1} H_i2^{r-1-i}\right)\).
\end{lemma}

\begin{proof}
Let \(A_i\) be the non-calendar cost incurred before the leader starts
waiting for the appropriate reset window, and let \(B_i\) be the
non-calendar cost incurred after that waiting period, including the reset
flooding. By the definition of \(H_i\), we may assume \(A_i+B_i\leq H_i\).
Let \(S_i=T_i+A_i\) be the elapsed time when the leader starts waiting.

By Lemma~\ref{lem:bc-reset-calendar-complexity}, the waiting period lasts at
most \(S_i+O(u_i+\log(S_i+2))\) rounds.
So the total time until the beginning of the next interval is
\(T_i + A_i + S_i + O(u_i+\log(S_i+2)) + B_i\), which, since \(S_i=T_i+A_i\),
equals \(2T_i + 2A_i + B_i + O(u_i+\log(S_i+2))\).
Using \(A_i+B_i\leq H_i\), \(A_i \leq H_i\) and \(H_i=\Omega(u_i)\), this is
\(2T_i + O(H_i+\log(T_i+H_i+2)) = 2T_i + O(H_i+\log(T_i+2))\).
Then \(T_{i+1} = 2T_i + O(H_i+\log(T_i+2))\).

We first derive a coarse bound. Since \(\log(T_i+2)\leq T_i+2\) and
\(H_i=\Omega(u_i)=\Omega(1)\), we get \(T_{i+1}\leq O(T_i + H_i)\). Since
the sequence \(H_i\) is nondecreasing and \(T_0=0\), induction gives
\(T_i\leq Kc^iH_i\) for some sufficiently large constant \(K\) and \(c\) independent
of \(i\).

We now bound the logarithmic term more sharply. From the coarse bound,
\(\log(T_i+2)\leq \log(Kc^iH_i+2)=O(i+\log(H_i+2))\). Since
\(H_i=\Omega(u_i)=\Omega(2^i)\), we have \(i=O(\log(H_i+2))\), and therefore
\(\log(T_i+2)=O(\log(H_i+2))=O(H_i)\). Substituting this into the previous
recurrence gives \(T_{i+1}\leq 2T_i+c'H_i\) for a constant \(c'\) independent
of \(i\).

Unfolding this recurrence over \(r\) resets gives
\[
    T_r
    =
    O\left(\sum_{i=0}^{r-1} H_i2^{r-1-i}\right),
\]
which proves the claimed bound.
\end{proof}

\section{Stabilizing \ProblemName{Input Multiset} and \ProblemName{Counting}}
\label{sec:stabilizing-input-multiset-counting}

We now combine the recovery layer with the known-bound \ProblemName{Input Multiset} algorithm. The procedure is given in Algorithm~\ref{alg:bc-stabilizing-input-multiset}.
Each trial simulates \InputMultiset{}.  Every flooding call is replaced by
\AdaptiveFlooding{}.  When a trial returns a candidate multiset, the agent
checks whether the total multiplicity represented by the candidate is larger
than \(\widehat U\).  If so, the estimate is certainly too small, and the agent
invalidates the trial by calling \DetectError{}.  Otherwise the candidate becomes
the current stabilizing output.

\begin{algorithm}[H]
\caption{Stabilizing \ProblemName{Input Multiset} without a known bound}
\label{alg:bc-stabilizing-input-multiset}

\KwInput{An input value $\mathit{input}$}
\KwOutput{A stabilizing output for the \ProblemName{Input Multiset} function}

\Fn{\AlgStartInputMultisetTrial{}}{
  initialize a simulation of \InputMultiset{} on $\mathit{input}$ \label{line:bc-sim-im-start}\;
  $\mathit{answer}\leftarrow\bot$\label{line:bc-sim-im-answer-bot}\;
}

\Fn{\AlgStabilizingInputMultiset{$\mathit{input}$}}{
  start the communication layer\label{line:bc-sim-im-start-layer}\;
  \AlgStartInputMultisetTrial{}\label{line:bc-sim-im-first-trial}\;
  \While{\boolname{true}}{
    \If{\Restart{} event occurs\label{line:bc-sim-im-restart-test}}{
      \AlgStartInputMultisetTrial{}\label{line:bc-sim-im-restart}\;
    }
    \Else{
      execute the next step of the current simulation, replacing each call to
      \Flood$(U-1,x)$ by \AdaptiveFlooding$(x)$\label{line:bc-sim-im-step}\;
      \If{the current simulation returns a candidate multiset $C$\label{line:bc-sim-im-candidate}}{
        \eIf{$C$ is not a well-formed multiset or its total multiplicity is larger than $\widehat U$\label{line:bc-sim-im-too-large}}{
          \AlgDetectError{}\label{line:bc-sim-im-detect}\;
        }{
          $\mathit{answer}\leftarrow C$\label{line:bc-sim-im-store}\;
        }
      }
    }
    output $\mathit{answer}$\label{line:bc-sim-im-output}\;
  }
}
\end{algorithm}

We use two simple properties of the simulated known-bound algorithm.  First, if a
trial starts at all agents in the same round and all completed calls to
\AdaptiveFlooding{} have returned the correct global OR synchronously, then the
next simulated flooding call, if any, is invoked by all agents in the same
round.  This is just the usual lockstep behavior of the known-bound synchronous
algorithm: all ordinary communication steps have fixed duration, and every
completed adaptive-flooding call has the same return round at all agents.
Second, the correctness proof of the known-bound \ProblemName{Input Multiset} algorithm uses
the bound only to guarantee that the flooding calls return the global OR.  Thus,
if all flooding calls are answered by correct global ORs, the simulated execution
returns the correct input multiset.

\begin{lemma}
\label{lem:bc-known-bound-simulation-interface}
Consider the known-bound \InputMultiset{} algorithm.  If each call
to \(\Flood(U-1,x)\) is replaced by a synchronous oracle that returns one bit,
and if all oracle calls return the correct global OR
at all agents in the same simulation-channel round, then a synchronously
started execution stays in lockstep.  The algorithm makes \(O(nB_{\max}+n^2\log n)\) such OR calls and
\(O(n)\) further one-round cut-test steps.
\end{lemma}

\begin{proof}
The known-bound algorithm is synchronous and deterministic.  Between flooding
calls, all ordinary communication steps have fixed duration, and every local
choice is determined by the common state, the local input, and the observations
already obtained.  Hence, if the previously consumed oracle calls return in the same
simulation-channel round at all agents, the next oracle call, if any, is
invoked by all agents in the same simulation-channel round.

The correctness proof of Theorem~\ref{thm:anonymous-bc-input-multiset} uses the
bound \(U\) only through Lemma~\ref{lem:flood-correctness}, namely to ensure that
each \(\Flood(U-1,x)\) call returns the global OR.  Replacing those calls by a
correct synchronous OR oracle therefore preserves the simulated execution and its
output.  The call bound follows from the complexity proof of
Theorem~\ref{thm:anonymous-bc-input-multiset}: the \(O(UnB_{\max}+Un^2\log n)\)
term is \(U\) times the number of flooding calls, while
\ConstraintRound{} contributes only \(O(n)\) additional one-round cut tests over
the whole execution.
\end{proof}

\begin{lemma}
\label{lem:bc-clean-adaptive-simulation}
Consider a synchronously started adaptive \ProblemName{Input Multiset} trial with common
estimate \(u\), and suppose all agents are active at the beginning of the
trial.  If no agent invalidates before the simulated execution produces a candidate
multiset \(C\), then all agents produce the same \(C\) in the same
simulation-channel round, and \(C\) is the correct input multiset.  In
particular, \(|C|=n\).
\end{lemma}

\begin{proof}
We argue by induction over the adaptive-flooding calls consumed by the
simulation before \(C\) is produced.  At the start of the trial, all agents
are in the same simulated state with the same estimate.  If all previous
adaptive-flooding calls have returned the correct global OR synchronously, the
lockstep part of Lemma~\ref{lem:bc-known-bound-simulation-interface} implies
that the next such call, if any, has a common start.  Since no agent invalidates
before \(C\) is produced and all active agents
have the same estimate by Lemma~\ref{lem:bc-common-estimate},
Lemma~\ref{lem:bc-adaptive-flooding-correctness} implies that this call returns
the correct global OR synchronously.  Thus every flooding answer consumed before
\(C\) is a correct global OR.

By Lemma~\ref{lem:bc-known-bound-simulation-interface}, the simulated known-bound execution therefore
returns the correct input multiset.  The same lockstep induction gives the same
return round and the same candidate \(C\) at all agents.  Since \(C\) is the
global input multiset, its mass is \(n\).
\end{proof}

\begin{lemma}[Progress of one \ProblemName{Input Multiset} trial]
\label{lem:bc-im-trial-progress}
Consider a synchronously started trial of Algorithm~\ref{alg:bc-stabilizing-input-multiset}
with estimate \(u\leq2n\), and suppose all agents are active at the beginning
of the trial.  Let $f(n)=n B_{\max}+n^2\log n$. If \(u<n\), then the trial becomes invalid within \(O(uf(n))\) rounds.  If
\(u\geq n\), then no agent becomes invalid, no further reset is initiated, and
all agents store the correct input multiset within \(O(uf(n))\) rounds.
\end{lemma}

\begin{proof}
By Lemma~\ref{lem:bc-known-bound-simulation-interface}, the known-bound \ProblemName{Input Multiset} algorithm
performs \(O(nB_{\max}+n^2\log n)=O(f(n))\) flooding calls and
synchronous steps.
  In the simulated trial, each flooding call is replaced by
\AdaptiveFlooding{} in line~\ref{line:bc-sim-im-step}.  By
Lemma~\ref{lem:bc-adaptive-flooding-cost}, each such replacement costs \(O(u)\)
rounds.  Hence, unless the trial is invalidated earlier, the simulated algorithm
returns a candidate multiset within \(O(uf(n))\) rounds.

If the trial is
invalidated earlier, then the first conclusion for \(u<n\) already holds within
this time bound.

Assume first that no agent invalidates before this return, and let \(C\) be the
candidate returned.  By Lemma~\ref{lem:bc-clean-adaptive-simulation}, \(C\) is
the correct input multiset and \(|C|=n\).  If \(u<n\), then \(C\) is not
\(u\)-valid, so
the test in line~\ref{line:bc-sim-im-too-large} succeeds and the agent calls
\DetectError{} in line~\ref{line:bc-sim-im-detect}.  Thus the trial is
invalidated within \(O(uf(n))\) rounds.  If \(u\geq n\), then \(C\) is \(u\)-valid, so the test in line~\ref{line:bc-sim-im-too-large} fails and the
correct candidate is stored in line~\ref{line:bc-sim-im-store}.

It remains, in the case \(u\geq n\), to check that no invalid propagation is created by
the layer before the candidate is stored.  The adaptive-flooding calls do not
call \DetectError{} in line~\ref{line:bc-af-detect}, by
Lemma~\ref{lem:bc-adaptive-flooding-correctness}.  The
invalidation channel cannot create an invalidated agent while none exists.  The
reset channel only reacts to an already started reset.  In the control channel,
the active leader starts the control signal at the beginning of each control
window in line~\ref{line:bc-control-start-signal} of Algorithm~\ref{alg:bc-error-control-rounds}; since the window length is
\(u\geq n\), the signal reaches every agent before the missing-signal test in
line~\ref{line:bc-control-missing} of Algorithm~\ref{alg:bc-error-control-rounds}, by the flooding argument of
Lemma~\ref{lem:flood-correctness}.  Hence no agent invalidates the trial.  Since a
reset can start only after the leader becomes invalid, no further reset is initiated.

\end{proof}

\begin{theorem}
\label{thm:bc-stabilizing-input-multiset}
\StabilizingInputMultiset{} stabilizes to the \ProblemName{Input Multiset} function.  More
precisely, after some finite time, every agent outputs the correct input
multiset forever.  From the beginning of the execution, the stabilization time is
$O(n^2B_{\max}\log n+n^3\log^2 n)$ rounds.
\end{theorem}

\begin{proof}

The case \(n=1\) is immediate, so assume \(n\geq2\).  Let
\(r=\lceil\log_2 n\rceil\) and \(u_i=2^i\).  The leader starts with estimate
\(u_0=1\) by line~\ref{line:bc-init-estimate}.  Whenever the leader initiates a
reset with estimate \(u_i\), the reset condition in
line~\ref{line:bc-reset-start-test} of Algorithm~\ref{alg:bc-reset-rounds} requires the reset-window value to be
\(2u_i=u_{i+1}\).  The leader is reached by the reset it starts, and installs
this value in line~\ref{line:bc-reset-update-estimate}.  Thus the leader's
estimates are \(u_0,u_1,u_2,\ldots\).

Consider an interval in which the leader's estimate is \(u_i<n\).  If the reset
that began this interval reached all agents, then all agents observe
\Restart{} synchronously by Lemma~\ref{lem:bc-error-reset-correctness}, and the
new trial is synchronously started with all agents active.  By
Lemma~\ref{lem:bc-im-trial-progress}, this trial is invalidated within
\(O(u_i f(n))\) rounds.  If the reset did not reach all agents, then some
unreached agent remained invalid by line~\ref{line:bc-reset-clear-invalid}; in
this case recovery is already active, and no progress guarantee for the simulated
trial is needed.  In either case, once an invalid agent exists, the
non-calendar recovery work before the next reset is \(O(n+u_i)\) rounds by
Lemma~\ref{lem:bc-recovery-cost}.  Hence the non-calendar cost charged to the
interval with leader estimate \(u_i\) is
\[
  H_i=K\bigl(u_i f(n)+n+u_i\bigr)
\]

for a sufficiently large absolute constant \(K\). The sequence \(H_i\) is
nondecreasing and satisfies \(H_i=\Omega(u_i)\), so Lemma~\ref{lem:bc-calendar-amortization} applies.

By Lemma~\ref{lem:bc-calendar-amortization}, the total time before the first
trial with leader estimate at least \(u_r\geq n\) is
\[
    O\left(\sum_{i=0}^{r-1} H_i2^{r-1-i}\right).
\]
Since \(2^{r-i}=\Theta(n/u_i)\), the simulation-work part contributes
\[
    \sum_{i=0}^{r-1} O\left(u_i f(n)\frac{n}{u_i}\right)
    =O(nf(n)\log n),
\]
and the recovery terms contribute
\[
    \sum_{i=0}^{r-1} O\left((n+u_i)\frac{n}{u_i}\right)=O(n^2).
\]

The first reset whose value is at least \(n\), if any, reaches all agents by
Lemma~\ref{lem:bc-error-reset-correctness}.  Therefore the following trial is
synchronously started with all agents active and estimate at least \(n\).  By
Lemma~\ref{lem:bc-im-trial-progress}, no further reset is initiated and all
agents store the correct input multiset within an additional
\(O(u_r f(n))=O(nf(n))\) rounds.
In the no-reset case, the initial estimate must already be sufficient for the
\ProblemName{Input Multiset} trial: otherwise Lemma~\ref{lem:bc-im-trial-progress} would
invalidate the trial and trigger recovery.  Hence the same lemma gives
stabilization in the initial trial.  Combining the bounds and substituting
\(f(n)=nB_{\max}+n^2\log n\) gives
\[
    O(nf(n)\log n+n^2)
    =O(n^2B_{\max}\log n+n^3\log^2 n),
\]
as claimed.
\end{proof}

\begin{corollary}[Stabilizing \ProblemName{Counting}]
\label{cor:bc-stabilizing-counting}
The same communication layer gives a stabilizing algorithm for the \ProblemName{Counting}
function.  From the beginning of the execution, it stabilizes in
\(O(n^3\log^2 n)\) rounds.
\end{corollary}

\begin{proof}
\ProblemName{Counting} is the special case of \ProblemName{Input Multiset} in which every agent uses the
same constant input value.  The final multiset then contains one value with
multiplicity \(n\), so every agent can output this multiplicity.  The
stabilization follows from Theorem~\ref{thm:bc-stabilizing-input-multiset}.  Since
the input alphabet has constant bit length in this reduction, the
\(O(n^2 B_{\max}\log n)\) term is absorbed by \(O(n^3\log^2 n)\).
\end{proof}

\subsection{Termination with a Common Bound}
\label{subsec:bc-terminating-common-bound}

Suppose that all agents are given a common bound \(U\geq n\).  We obtain a
terminating algorithm by running the stabilizing algorithm and delaying
acceptance of each candidate output.  When an active agent obtains a
\(\widehat U\)-valid candidate multiset \(C\), it stores \(C\), stops advancing
the simulation, and waits for \(U+1\) simulation-channel rounds.  If it invalidates
or receives a \Restart{} event during this waiting period, it discards \(C\) and
continues with the next trial.  Otherwise, when the waiting period ends, it
terminates and outputs \(C\).  As before, candidates that are not
\(\widehat U\)-valid are rejected by calling \DetectError{}.

\begin{theorem}
\label{thm:bc-terminating-input-multiset}
In anonymous \(1\)-interval-connected dynamic networks with a unique leader,
given a common bound \(U\geq n\), the \ProblemName{Input Multiset} function can be computed
with termination in $O(n^2B_{\max}\log n+n^3\log^2 n+U)$ rounds.
\end{theorem}

\begin{proof}
By Theorem~\ref{thm:bc-stabilizing-input-multiset}, within $O(n^2B_{\max}\log n+n^3\log^2 n)$ rounds all later candidates produced by the stabilizing algorithm are the
correct input multiset and no further restart is generated.  The next such
candidate is therefore correct and is followed by a waiting period of
\(U+1=O(U)\) simulation-channel rounds, hence \(O(U)\) ordinary rounds.

It remains only to rule out termination with an earlier incorrect candidate.  If
some agent has already invalidated when such a candidate is stored, then during
the \(U+1\) simulation-channel waiting period at least \(U\geq n\) invalidation-channel
rounds occur.  Hence the invalidation reaches the agent before the waiting period can
finish, unless a reset completion reaches it first; in either case the candidate
is discarded.  If no agent has become invalid when the candidate is stored, then the
trial was synchronously started and Lemma~\ref{lem:bc-clean-adaptive-simulation}
implies that the candidate is correct.  Thus no incorrect candidate can be
accepted, and the stated bound follows.
\end{proof}

\begin{corollary}
\label{cor:bc-terminating-counting}
In anonymous \(1\)-interval-connected dynamic networks with a unique leader,
given a common bound \(U\geq n\), the \ProblemName{Counting} function can be computed with
termination in $O(n^3\log^2 n+U)$
rounds.
\end{corollary}

\color{black}

\section{Related Work}
\label{sec:related-work}

We survey work on anonymous networks and on communication models close to
the one-bit model studied in this paper.

\paragraph{Anonymous networks.}
The study of computability in anonymous static networks goes back to
Angluin~\cite{angluin_STOC1980} and has developed into a broad line of
work~\cite{boldi_vigna_DISC2001,chalopin_godard_metivier_DISC2008,chalopin_metivier_morsellino_FI2012,fraigniaud_pelc_peleg_perennes_PODC2000,yamashita_kameda_PODC1988,yamashita_kameda_TPDS1996}.
The main obstacle introduced by anonymity is symmetry: agents with the
same view of the network cannot be distinguished by a deterministic
algorithm. For instance, in an anonymous static ring, all agents have
the same view, and this view does not by itself distinguish rings of
different sizes. Hence non-trivial tasks such as counting require
additional information or a symmetry-breaking assumption, such as the
presence of a leader.

Several works provide structural characterizations of computability in
anonymous networks. Yamashita and Kameda characterized the tasks solvable
in anonymous networks of known size through the combinatorial notion of
views~\cite{yamashita_kameda_TPDS1996}. Boldi and Vigna gave a
characterization of computability and stabilization using graph fibrations
and coverings~\cite{boldi_vigna_DISC2001}. The above techniques only work on static networks and use messages of size at least $\log (n)$.

\paragraph{Anonymous dynamic networks.}
The problem of counting and computing in anonymous dynamic networks has been studied in several papers~\cite{michail_chatzigiannakis_spirakis_SSS2013,kowalski_mosteiro_ICALP2018,kowalski_mosteiro_JCSS2021} that have given polynomial algorithms for counting and related problems.

For anonymous \(1\)-interval-connected dynamic networks, Di Luna and
Viglietta, introducing the history-tree technique, showed that every computable function, including the one studied in our paper, can be computed in linear
time in the presence of a leader~\cite{diluna_viglietta_FOCS2022}.
Subsequent work extended their history-tree approach to the leaderless and
multi-leader setting~\cite{diluna_viglietta_DISC2023_disconnected}, to
finite-state and self-stabilizing computation~\cite{diluna_viglietta_OPODIS2024},
and to congested networks~\cite{diluna_viglietta_DC2025_congested}.

These history-tree algorithms rely on agents transmitting at least
parts of their vistas. In the congested model, where messages
have \(O(\log n)\) bits, this yields an \(O(n^3)\)-round algorithm for computing \ProblemName{Input Multiset} when a
leader is present~\cite{diluna_viglietta_DC2025_congested}. Our paper studies
the stricter setting in which each agent broadcasts only a single bit
per round; the
history-tree approach used in the previous algorithms therefore does
not directly apply.

\paragraph{Beeping networks.}
The beeping model was introduced as a very weak radio communication model by
Cornejo and Kuhn~\cite{cornejo_kuhn_DISC2010} and, to our knowledge, all works have examined static networks.

Metivier, Robson, and Zemmari
then studied the algorithmic power of several beeping variants, including
collision-detection emulations and local graph problems~\cite{metivier_robson_zemmari_arxiv2015}.
For counting, Casteigts, Metivier, Robson, and Zemmari considered one-hop
beeping networks~\cite{casteigts_metivier_robson_zemmari_TCS2019}. In the beeping model, exact deterministic counting cannot be done, and
the impossibility even holds for Las Vegas algorithms; with sender-side
collision detection, they give an $O(n)$ randomized algorithm with high
probability. Brandes, Kardas, Klonowski, Pajak, and Wattenhofer studied
randomized size approximation in one-hop beeping networks and obtained
time-optimal approximation bounds~\cite{brandes_et_al_TCS2020}. These results concern one-hop networks and are randomized, while our results are deterministic in
multi-hop dynamic networks.

Regarding multi-hop static networks,
Czumaj and Davies developed deterministic beeping algorithms for static global
communication tasks including broadcast, gossiping, and
multi-broadcast~\cite{czumaj_davies_OPODIS2015,czumaj_davies_JPDC2019},
and Beauquier, Burman, Davies, and Dufoulon later obtained uniform
time-optimal deterministic multi-broadcast via group
testing~\cite{beauquier_burman_davies_dufoulon_SIROCCO2019}. Davies obtained
an \(O(\Delta \log n)\)-round simulation of one Broadcast-CONGEST round and
an \(O(\Delta^2 \log n)\)-round simulation of one CONGEST round in the noisy
beeping model~\cite{davies_DC2025_noisy}.
These works use techniques that rely on fixed neighborhoods, identifiers, or degree/diameter
structure, and therefore do not directly transfer to anonymous
\(1\)-interval-connected dynamic networks.

\paragraph{Relationship between our model and beeping networks.} The communication model studied in this paper is closely related to beeping
models. In both settings, each agent chooses between two possible actions in
each round: in beeping models these actions are usually called \emph{beep} and
\emph{silence} or \emph{listen}, while in our model they are denoted by \(1\) and \(0\). 

The basic beeping model is half-duplex. An agent that emits a signal receives
no information in that round, while an agent that remains silent only learns
whether at least one neighbor emitted a signal. Thus the feedback is both
asymmetric and boolean: it distinguishes zero from nonzero, but does not
provide multiplicities.

In counting full-duplex variants of beeping an agent learns, in every round,
how many of its neighbors emitted a signal.
However, this is still weaker than our model, since it does not
learn how many neighbors remained silent. If, in addition, the agent can observe its degree after choosing its action,
then it can recover the number of silent neighbors. This gives
exactly the same information as in the broadcast-counting model.

Consequently, the algorithms presented in this paper also apply to the
corresponding counting beeping models under these additional assumptions.

\paragraph{Content-oblivious model.}
The content-oblivious (CO) model was introduced by Censor-Hillel, Cohen, Gelles, and
Sela. In this model agents communicate
over an asynchronous network by exchanging pulses~\cite{censor_hillel_cohen_gelles_sela_DC2023}.
A pulse carries no content, not even the identity of the sender, so an
agent can only use the number of pulses received on each local port. From this perspective a pulse is weaker than a message carrying a single bit.
To the best of our knowledge, all work on this model has examined static networks \cite{censor_hillel_cohen_gelles_sela_DC2023,frei_et_al_DISC2024,chalopin_chang_chen_diluna_zhou_arxiv2025_nonuniform,chalopin_chang_chen_diluna_zhou_arxiv2025_2edge,chang_chen_zhou_arxiv2025_beyond,chalopin_chang_diluna_zhou_arxiv2026}.

Censor-Hillel et al. showed that, in the presence of a leader, CO
communication can simulate asynchronous message passing on 2-edge-connected
networks~\cite{censor_hillel_cohen_gelles_sela_DC2023}. Subsequent work
removed or weakened the leader requirement in several settings: Frei,
Gelles, Ghazy, and Nolin gave a leader-election algorithm for oriented rings
\cite{frei_et_al_DISC2024}; Chalopin, Chang, Chen, Di Luna, and Zhou
gave a non-uniform leader-election algorithm on unoriented rings and a leader-election algorithm
in uniform 2-edge-connected networks
\cite{chalopin_chang_chen_diluna_zhou_arxiv2025_nonuniform,chalopin_chang_chen_diluna_zhou_arxiv2025_2edge};
while Chang, Chen, and Zhou studied what remains possible beyond
2-edge-connectivity~\cite{chang_chen_zhou_arxiv2025_beyond}.

Closest to our counting problem, Chalopin, Chang, Di Luna, and Zhou study
counting and simulation in content-oblivious rings
\cite{chalopin_chang_diluna_zhou_arxiv2026}. They give an
\(O(n^{3/2})\)-pulse counting algorithm for anonymous rings with a leader,
an \(O(n\log^2 n)\)-pulse algorithm for rings with identifiers, and an
\(\Omega(n\log n)\) lower bound for counting in the CO model. They also
give a simulator for asynchronous message passing on general 2-edge-connected networks with constant pulse
overhead per transmitted bit after preprocessing. The result on counting in an anonymous ring with a leader is close
in spirit to ours, but the models are different: their algorithm is
asynchronous and ring-based, and agents can send a message over a single link, while our algorithms are synchronous, broadcast-based, and work
in arbitrary \(1\)-interval-connected dynamic networks. The dynamic topology and broadcast communication of our setting do not allow the transfer of the content-oblivious algorithmic techniques.

\newcommand{\etalchar}[1]{$^{#1}$}


\begin{thebibliography}{99}

\bibitem{andriambolamalala_ravelomanana_arxiv2021}
Ny~Aina Andriambolamalala and Vlady Ravelomanana.
\newblock {Energy-Efficient Naming in Beeping Networks}, 2021.
\newblock URL: \url{https://arxiv.org/abs/2106.03753},
  \href{https://arxiv.org/abs/2106.03753}{\path{arXiv:2106.03753}}.

\bibitem{angluin_STOC1980}
Dana Angluin.
\newblock {Local and Global Properties in Networks of Processors (Extended
  Abstract)}.
\newblock In {\em Proceedings of the 12th Annual ACM Symposium on Theory of
  Computing (STOC 1980)}, pages 82--93. ACM, 1980.
\newblock \href{https://doi.org/10.1145/800141.804655}
  {\path{doi:10.1145/800141.804655}}.

\bibitem{beauquier_burman_davies_dufoulon_SIROCCO2019}
Joffroy Beauquier, Janna Burman, Peter Davies, and Fabien Dufoulon.
\newblock {Optimal Multi-broadcast with Beeps Using Group Testing}.
\newblock In {\em Structural Information and Communication Complexity
  (SIROCCO 2019)}, volume 11639 of {\em Lecture Notes in Computer Science},
  pages 66--80. Springer, 2019.
\newblock \href{https://doi.org/10.1007/978-3-030-24922-9_5}
  {\path{doi:10.1007/978-3-030-24922-9_5}}.

\bibitem{boldi_vigna_DISC2001}
Paolo Boldi and Sebastiano Vigna.
\newblock {An Effective Characterization of Computability in Anonymous
  Networks}.
\newblock In {\em Distributed Computing}, volume 2180 of {\em Lecture Notes in
  Computer Science}, pages 33--47. Springer, 2001.
\newblock \href{https://doi.org/10.1007/3-540-45414-4_3}
  {\path{doi:10.1007/3-540-45414-4_3}}.

\bibitem{brandes_et_al_TCS2020}
Philipp Brandes, Marcin Kardas, Marek Klonowski, Dominik Pajak, and Roger
  Wattenhofer.
\newblock {Fast Size Approximation of a Radio Network in Beeping Model}.
\newblock {\em Theoretical Computer Science}, 810:15--25, 2020.
\newblock \href{https://doi.org/10.1016/j.tcs.2017.05.022}
  {\path{doi:10.1016/j.tcs.2017.05.022}}.

\bibitem{casteigts_metivier_robson_zemmari_TCS2019}
Arnaud Casteigts, Yves M{\'e}tivier, John~Michael Robson, and Akka Zemmari.
\newblock {Counting in One-Hop Beeping Networks}.
\newblock {\em Theoretical Computer Science}, 780:20--28, 2019.
\newblock \href{https://doi.org/10.1016/j.tcs.2019.02.009}
  {\path{doi:10.1016/j.tcs.2019.02.009}}.

\bibitem{censor_hillel_cohen_gelles_sela_DC2023}
Keren Censor-Hillel, Shir Cohen, Ran Gelles, and Gal Sela.
\newblock {Distributed Computations in Fully-Defective Networks}.
\newblock {\em Distributed Computing}, 36(4):501--528, 2023.
\newblock \href{https://doi.org/10.1007/s00446-023-00452-2}
  {\path{doi:10.1007/s00446-023-00452-2}}.

\bibitem{chalopin_chang_chen_diluna_zhou_arxiv2025_2edge}
J{\'e}r{\'e}mie Chalopin, Yi-Jun Chang, Lyuting Chen, Giuseppe~A. Di~Luna,
  and Haoran Zhou.
\newblock {Content-Oblivious Leader Election in 2-Edge-Connected Networks},
  2025.
\newblock URL: \url{https://arxiv.org/abs/2507.08348},
  \href{https://arxiv.org/abs/2507.08348}{\path{arXiv:2507.08348}}.

\bibitem{chalopin_chang_chen_diluna_zhou_arxiv2025_nonuniform}
J{\'e}r{\'e}mie Chalopin, Yi-Jun Chang, Lyuting Chen, Giuseppe~A. Di~Luna,
  and Haoran Zhou.
\newblock {Non-Uniform Content-Oblivious Leader Election on Oriented
  Asynchronous Rings}, 2025.
\newblock URL: \url{https://arxiv.org/abs/2509.19187},
  \href{https://arxiv.org/abs/2509.19187}{\path{arXiv:2509.19187}}.

\bibitem{chalopin_chang_diluna_zhou_arxiv2026}
J{\'e}r{\'e}mie Chalopin, Yi-Jun Chang, Giuseppe~A. Di~Luna, and Haoran Zhou.
\newblock {Efficient Counting and Simulation in Content-Oblivious Rings}.
\newblock To appear in {\em Proceedings of the 45th ACM Symposium on
  Principles of Distributed Computing (PODC 2026)}. ACM, 2026.
\newblock URL: \url{https://arxiv.org/abs/2603.28260},
  \href{https://arxiv.org/abs/2603.28260}{\path{arXiv:2603.28260}}.

\bibitem{chalopin_godard_metivier_DISC2008}
J{\'e}r{\'e}mie Chalopin, Emmanuel Godard, and Yves M{\'e}tivier.
\newblock {Local Terminations and Distributed Computability in Anonymous
  Networks}.
\newblock In {\em Distributed Computing}, volume 5218 of {\em Lecture Notes in
  Computer Science}, pages 47--62. Springer, 2008.
\newblock \href{https://doi.org/10.1007/978-3-540-87779-0_4}
  {\path{doi:10.1007/978-3-540-87779-0_4}}.

\bibitem{chalopin_metivier_morsellino_FI2012}
J{\'e}r{\'e}mie Chalopin, Yves M{\'e}tivier, and Thomas Morsellino.
\newblock {Enumeration and Leader Election in Partially Anonymous and
  Multi-hop Broadcast Networks}.
\newblock {\em Fundamenta Informaticae}, 120(1):1--27, 2012.
\newblock \href{https://doi.org/10.3233/FI-2012-747}
  {\path{doi:10.3233/FI-2012-747}}.

\bibitem{chang_chen_zhou_arxiv2025_beyond}
Yi-Jun Chang, Lyuting Chen, and Haoran Zhou.
\newblock {Beyond 2-Edge-Connectivity: Algorithms and Impossibility for
  Content-Oblivious Leader Election}, 2025.
\newblock URL: \url{https://arxiv.org/abs/2511.23297},
  \href{https://arxiv.org/abs/2511.23297}{\path{arXiv:2511.23297}}.

\bibitem{chlebus_demarco_talo_FI2017}
Bogdan~S. Chlebus, Gianluca De~Marco, and Muhammed Talo.
\newblock {Naming a Channel with Beeps}.
\newblock {\em Fundamenta Informaticae}, 153(3):199--219, 2017.
\newblock \href{https://doi.org/10.3233/FI-2017-1537}
  {\path{doi:10.3233/FI-2017-1537}}.

\bibitem{cornejo_kuhn_DISC2010}
Alejandro Cornejo and Fabian Kuhn.
\newblock {Deploying Wireless Networks with Beeps}.
\newblock In {\em Distributed Computing}, volume 6343 of {\em Lecture Notes in
  Computer Science}, pages 148--162. Springer, 2010.
\newblock URL: \url{https://arxiv.org/abs/1005.2567}.

\bibitem{czumaj_davies_OPODIS2015}
Artur Czumaj and Peter Davies.
\newblock {Communicating with Beeps}.
\newblock In {\em 19th International Conference on Principles of Distributed
  Systems (OPODIS 2015)}, volume 46 of {\em LIPIcs}, pages 30:1--30:16.
  Schloss Dagstuhl -- Leibniz-Zentrum f{\"u}r Informatik, 2016.
\newblock \href{https://doi.org/10.4230/LIPIcs.OPODIS.2015.30}
  {\path{doi:10.4230/LIPIcs.OPODIS.2015.30}}.

\bibitem{czumaj_davies_JPDC2019}
Artur Czumaj and Peter Davies.
\newblock {Communicating with Beeps}.
\newblock {\em Journal of Parallel and Distributed Computing}, 130:98--109,
  2019.
\newblock \href{https://doi.org/10.1016/j.jpdc.2019.03.020}
  {\path{doi:10.1016/j.jpdc.2019.03.020}}.

\bibitem{davies_DC2025_noisy}
Peter Davies-Peck.
\newblock {Optimal Message-Passing with Noisy Beeps}.
\newblock {\em Distributed Computing}, 38(3):247--260, 2025.
\newblock \href{https://doi.org/10.1007/s00446-025-00488-6}
  {\path{doi:10.1007/s00446-025-00488-6}}.

\bibitem{diluna_baldoni_bonomi_chatzigiannakis_ICDCS2014}
Giuseppe~A. Di~Luna, Roberto Baldoni, Silvia Bonomi, and Ioannis
  Chatzigiannakis.
\newblock {Counting in Anonymous Dynamic Networks under Worst-Case Adversary}.
\newblock In {\em 2014 IEEE 34th International Conference on Distributed
  Computing Systems (ICDCS 2014)}, pages 338--347. IEEE, 2014.
\newblock \href{https://doi.org/10.1109/ICDCS.2014.42}
  {\path{doi:10.1109/ICDCS.2014.42}}.

\bibitem{diluna_viglietta_FOCS2022}
Giuseppe~A. Di~Luna and Giovanni Viglietta.
\newblock {Computing in Anonymous Dynamic Networks Is Linear}.
\newblock In {\em 2022 IEEE 63rd Annual Symposium on Foundations of Computer
  Science (FOCS 2022)}, pages 1122--1133. IEEE, 2022.
\newblock \href{https://doi.org/10.1109/FOCS54457.2022.00108}
  {\path{doi:10.1109/FOCS54457.2022.00108}}.

\bibitem{diluna_viglietta_DISC2023_disconnected}
Giuseppe~A. Di~Luna and Giovanni Viglietta.
\newblock {Optimal Computation in Leaderless and Multi-Leader Disconnected
  Anonymous Dynamic Networks}.
\newblock In {\em 37th International Symposium on Distributed Computing
  (DISC 2023)}, volume 281 of {\em LIPIcs}, pages 18:1--18:20.
  Schloss Dagstuhl -- Leibniz-Zentrum f{\"u}r Informatik, 2023.
\newblock \href{https://doi.org/10.4230/LIPIcs.DISC.2023.18}
  {\path{doi:10.4230/LIPIcs.DISC.2023.18}}.

\bibitem{diluna_viglietta_OPODIS2024}
Giuseppe~A. Di~Luna and Giovanni Viglietta.
\newblock {Universal Finite-State and Self-Stabilizing Computation in
  Anonymous Dynamic Networks}.
\newblock In {\em 28th International Conference on Principles of Distributed
  Systems (OPODIS 2024)}, volume 324 of {\em LIPIcs}, pages 10:1--10:17.
  Schloss Dagstuhl -- Leibniz-Zentrum f{\"u}r Informatik, 2024.
\newblock \href{https://doi.org/10.4230/LIPIcs.OPODIS.2024.10}
  {\path{doi:10.4230/LIPIcs.OPODIS.2024.10}}.

\bibitem{diluna_viglietta_DC2025_congested}
Giuseppe~A. Di~Luna and Giovanni Viglietta.
\newblock {Efficient Computation in Congested Anonymous Dynamic Networks}.
\newblock {\em Distributed Computing}, 38(2):95--112, 2025.
\newblock URL: \url{https://arxiv.org/abs/2301.07849}.

\bibitem{fraigniaud_pelc_peleg_perennes_PODC2000}
Pierre Fraigniaud, Andrzej Pelc, David Peleg, and St{\'e}phane P{\'e}rennes.
\newblock {Assigning Labels in Unknown Anonymous Networks (Extended Abstract)}.
\newblock In {\em Proceedings of the 19th Annual ACM Symposium on Principles of
  Distributed Computing (PODC 2000)}, pages 101--111. ACM, 2000.
\newblock \href{https://doi.org/10.1145/343477.343527}
  {\path{doi:10.1145/343477.343527}}.

\bibitem{frei_et_al_DISC2024}
Fabian Frei, Ran Gelles, Ahmed Ghazy, and Alexandre Nolin.
\newblock {Content-Oblivious Leader Election on Rings}.
\newblock In {\em 38th International Symposium on Distributed Computing
  (DISC 2024)}, volume 319 of {\em LIPIcs}, pages 26:1--26:20.
  Schloss Dagstuhl -- Leibniz-Zentrum f{\"u}r Informatik, 2024.
\newblock \href{https://doi.org/10.4230/LIPIcs.DISC.2024.26}
  {\path{doi:10.4230/LIPIcs.DISC.2024.26}}.

\bibitem{garncarek_et_al_ESA2025}
Pawe{\l} Garncarek, Dariusz~R. Kowalski, Shay Kutten, and Miguel~A. Mosteiro.
\newblock {Beeping Deterministic CONGEST Algorithms in Graphs}.
\newblock In {\em 33rd Annual European Symposium on Algorithms (ESA 2025)},
  volume 351 of {\em LIPIcs}, pages 20:1--20:17. Schloss Dagstuhl --
  Leibniz-Zentrum f{\"u}r Informatik, 2025.
\newblock \href{https://doi.org/10.4230/LIPIcs.ESA.2025.20}
  {\path{doi:10.4230/LIPIcs.ESA.2025.20}}.

\bibitem{kowalski_mosteiro_ICALP2018}
Dariusz~R. Kowalski and Miguel~A. Mosteiro.
\newblock {Polynomial Counting in Anonymous Dynamic Networks with Applications
  to Anonymous Dynamic Algebraic Computations}.
\newblock In {\em 45th International Colloquium on Automata, Languages, and
  Programming (ICALP 2018)}, volume 107 of {\em LIPIcs}, pages 156:1--156:14.
  Schloss Dagstuhl -- Leibniz-Zentrum f{\"u}r Informatik, 2018.
\newblock \href{https://doi.org/10.4230/LIPIcs.ICALP.2018.156}
  {\path{doi:10.4230/LIPIcs.ICALP.2018.156}}.

\bibitem{kowalski_mosteiro_JCSS2021}
Dariusz~R. Kowalski and Miguel~A. Mosteiro.
\newblock {Polynomial Anonymous Dynamic Distributed Computing without a Unique
  Leader}.
\newblock {\em Journal of Computer and System Sciences}, 123:37--63, 2022.
\newblock \href{https://doi.org/10.1016/j.jcss.2021.07.002}
  {\path{doi:10.1016/j.jcss.2021.07.002}}.

\bibitem{lynch1996distributed}
Nancy~A. Lynch.
\newblock {\em Distributed Algorithms}.
\newblock Morgan Kaufmann Publishers, San Francisco, CA, 1996.

\bibitem{metivier_robson_zemmari_arxiv2015}
Yves M{\'e}tivier, John~Michael Robson, and Akka Zemmari.
\newblock {On Distributed Computing with Beeps}, 2015.
\newblock URL: \url{https://arxiv.org/abs/1507.02721},
  \href{https://arxiv.org/abs/1507.02721}{\path{arXiv:1507.02721}}.

\bibitem{michail_chatzigiannakis_spirakis_SSS2013}
Othon Michail, Ioannis Chatzigiannakis, and Paul~G. Spirakis.
\newblock {Naming and Counting in Anonymous Unknown Dynamic Networks}.
\newblock In {\em Stabilization, Safety, and Security of Distributed Systems},
  volume 8255 of {\em Lecture Notes in Computer Science}, pages 281--295.
  Springer, 2013.
\newblock \href{https://doi.org/10.1007/978-3-319-03089-0_20}
  {\path{doi:10.1007/978-3-319-03089-0_20}}.

\bibitem{olshevsky_SICON2017}
Alex~Olshevsky.
\newblock {Linear Time Average Consensus and Distributed Optimization on Fixed
  Graphs}.
\newblock {\em SIAM Journal on Control and Optimization}, 55(6):3990--4014,
  2017.

\bibitem{viglietta_SIROCCO2024}
Giovanni Viglietta.
\newblock {History Trees and Their Applications}.
\newblock In {\em 31st International Colloquium on Structural
  Information and Communication Complexity (SIROCCO 2024)}, volume 14662 of {\em LNCS}, pages 3--23, 2024.
\newblock URL: \url{https://arxiv.org/abs/2404.02673},
  \href{https://arxiv.org/abs/2404.02673}{\path{arXiv:2404.02673}}.

\bibitem{yamashita_kameda_PODC1988}
Masafumi Yamashita and Tsunehiko Kameda.
\newblock {Computing on an Anonymous Network}.
\newblock In {\em Proceedings of the 7th Annual ACM Symposium on Principles of
  Distributed Computing (PODC 1988)}, pages 117--130. ACM, 1988.

\bibitem{yamashita_kameda_TPDS1996}
Masafumi Yamashita and Tsunehiko Kameda.
\newblock {Computing on Anonymous Networks. I. Characterizing the Solvable
  Cases}.
\newblock {\em IEEE Transactions on Parallel and Distributed Systems},
  7(1):69--89, 1996.
\newblock \href{https://doi.org/10.1109/71.481599}
  {\path{doi:10.1109/71.481599}}.

\end{thebibliography}
\end{document}